\documentclass{article}

\usepackage{setspace}
\usepackage{color,soul}
\usepackage{amsfonts}
\usepackage{dsfont}
\usepackage[cmex10]{amsmath}
\usepackage{amssymb}
\usepackage{multirow}

\usepackage[table]{xcolor}

\usepackage{algorithm}%
\usepackage{algpseudocode}%
\algnewcommand{\algorithmicgoto}{\textbf{go to}}%
\algnewcommand{\Goto}[1]{\algorithmicgoto~\ref{#1}}%

\usepackage[mathscr]{euscript}

\usepackage{graphicx}
\graphicspath{ {./} }
\DeclareGraphicsExtensions{.pdf}

\usepackage{subcaption}

\newtheorem{definition}{Definition}
\newtheorem{theorem}{Theorem}%
\newtheorem{example}{Example}%
\newtheorem{remark}{Remark}%
\newtheorem{proof}{Proof}%
\newtheorem{proposition}{Proposition}%

\usepackage{framed} %

\usepackage{geometry} %

\newenvironment{keywords}{
	\vspace{1em}\noindent\textbf{Keywords:}\ } %
{}

\begin{document}

\newgeometry{top=72pt,bottom=54pt,right=54pt,left=54pt}

\title{Fastest Mixing Reversible Markov Chain: Clique Lifted Graphs \& Subgraphs}
\author{Saber Jafarizadeh \\
	Member~IEEE \\
	Rakuten Institute of Technology, Rakuten Crimson House, Tokyo, Japan \\
	\texttt{saber.jafarizadeh@rakuten.com}}

\date{}
\maketitle
\thispagestyle{empty} %

\bibliographystyle{plain}

\begin{abstract}

Markov chains are one of the well-known tools for modeling and analyzing stochastic systems. 
At the same time, they are used for constructing random walks that can achieve a given stationary distribution.
This paper is concerned with determining the transition probabilities that optimize the mixing time of the reversible Markov chains towards a given equilibrium distribution.
This problem is referred to as the Fastest Mixing Reversible Markov Chain (FMRMC) problem.
It
is shown that for a given base graph and its clique lifted graph, the FMRMC problem over the clique lifted graph is reducible to the FMRMC problem over the base graph, while the optimal
mixing times
on both graphs are identical.
Based on this result and the solution of the semidefinite programming formulation of the FMRMC problem, the problem has been addressed over a wide variety of topologies with the same base graph.
Second, the general form of the FMRMC problem is addressed on stand-alone topologies as well as subgraphs of an arbitrary graph.
For subgraphs, it is shown that the optimal transition probabilities over edges of the subgraph can be determined independent of rest of the topology.

\end{abstract}

\begin{keywords}
Reversible Markov chains, Lifted graphs, Optimal mixing time, Semidefinite Programming
\end{keywords}

\section{Introduction}
\label{sec:Introduction}

Markov chains have been the focus of many studies due to their application in different fields including biology, chemistry, physics, statistics and computer science \cite{7492191,Seismic2015,7469382,Metropolis1953,Kirkpatrick1983}.
Within the context of Markov Chain Monte Carlo methods, reversible Markov chains \cite{AldousBook2} have been used for sampling from a given probability distribution \cite{Stutzbach2009,Gjoka2011}.
Also, in networking applications such as network storage \cite{Lin2007}, information dissemination \cite{Vukobratovic2010}, randomized search \cite{Zhong2006} and random membership subset management \cite{Zhong2008},
reversible Markov chains are used for constructing random walks on a graph that can achieve a given stationary distribution.

An important parameter of a Markov chain is its mixing rate and its optimization problem is known as the Fastest Mixing Markov Chain (FMMC) problem.
The main goal, in the FMMC problem is to design the transition probabilities of the Markov chain such that its convergence towards the (given) target equilibrium distribution is optimized, while the transition probabilities are restricted according to the topology of network.
In \cite{BoydFastestmixing2003}, it is shown that the asymptotic convergence rate of the Markov chain to the equilibrium distribution is determined by the Second Largest Eigenvalue Modulus ($SLEM$) of the transition probability matrix.
Different approaches have been adopted to solve the FMMC problem.
In \cite{Lawler1988,Diaconis1991}, the spectrum of Markov Chains has been bounded.
Authors in \cite{BoydFastestmixing2003} have formulated the FMMC problem as a Semidefinite Program (SDP), where
for the case of uniform equilibrium distribution, symmetry of the underlying topology has been exploited for solving the SDP formulation of the FMMC problem \cite{BoydFastestMixing2009,Boyd2005SymmetryAnalysisFMMC,Boyd2006FMMCPath,Jafarizadeh2011,SaberThesis2015,Allisona2011}.
Authors in \cite{Fill2013} have introduced stochastically monotone Markov kernels and by taking advantage of the monotonicity of mixing times, they have solved the birth and death problem.
\cite{Cihan2015} solves the FMMC problem for the special case of edge-transitive and distance-transitive graphs where the equilibrium distribution is proportional to the degree of vertices.

This paper studies the FMMC problem for reversible Markov chains
(referred to as the FMRMC problem)
on
undirected
graphs with given equilibrium distribution which is not necessarily uniform.
We have addressed the FMRMC problem in its general form, by rewriting the original optimization problem in terms of the symmetric Laplacian, and deriving its corresponding SDP problem.

Following are the main contributions of the this paper.

\hspace{-13pt} $\bullet$
We have shown that
for a given base graph and its clique lifted graph,
the FMRMC problem
(with any given equilibrium distribution)
over the clique lifted graph
is reducible to
the FMRMC problem over the base graph,
while the optimal $SLEM$ of the clique lifted graph (and thus its optimal convergence rate) is equal to that of the base graph.
This has enabled us to address the FMRMC problem over a wide variety of topologies which share a common base graph.
We have also shown that tight interlacing holds between the eigenvalues of the optimal transition probability matrices corresponding to a clique lifted graph and its base graph.

\hspace{-13pt} $\bullet$
For
families
of
clique lifted
subgraphs
with path subgraph as their base graph,
given that it is connected to one vertex in the rest of the topology,
we have provided the optimal transition probabilities of the FMRMC problem, in the subgraph
and we have shown that they are obtained locally, i.e. independent of the rest of the topology.
Similar results have been achieved for palm subgraph, even though it is not a clique lifted subgraph,
and interestingly,
it is shown that the optimal transition probabilities from every vertex in the star part of the subgraphs to the central vertex are all equal.

\hspace{-13pt} $\bullet$
By addressing the FMRMC problem over different types of topologies with path graph as their base graph and with arbitrary (not necessarily symmetric) equilibrium distribution, we have determined the optimal transition probabilities and provided the algorithm for calculating the optimal $SLEM$ over these topologies.
For topologies, other than Path (with $3$ vertices) and Star topologies, the optimal results are presented for the cases where the optimal answer is obtained from the interior of the domain.
Based on the results provided for star topology, the FMRMC problem for all equilibrium distributions has been addressed over both windmill and friendship graphs.
Since both of them are clique lifted of star graph.

\hspace{-13pt} $\bullet$
Furthermore, we have considered the FMRMC problem over symmetric tree, symmetric star and complete cored symmetric star topologies assuming that the symmetric pattern of the equilibrium distribution is maintained.
By exploiting their symmetry, we have determined the optimal transition probabilities and provided the algorithm for calculating the optimal $SLEM$ over these topologies.
The results presented for symmetric tree topology of depth two cover the whole domain.
Furthermore,
in the case of symmetric tree topology with arbitrary depth,
we have been provided the
optimal transition probabilities
for the case where the diagonal elements of transition probability matrix are zero.

\hspace{-13pt} $\bullet$
Via numerical examples, for every studied topology and subgraph, we have compared the optimal transition probabilities with Metropolis transition probabilities.
This comparison has been done in terms of convergence rate.

General description of the FMRMC problem has been introduced in Section \ref{sec:FMRMCsection}.
In section \ref{sec:FMRMCSymmetricLaplacian}, the FMMC problem has been formulated as an optimization problem with symmetric optimization variables (i.e. the symmetric Laplacian).
Section \ref{sec:FMRMCCliqueLiftGraphs} addresses the FMRMC problem over clique lifted graphs, where its semidefinite programming formulation has been presented in Section \ref{sec:FMRMCCliqueLiftGraphsSDP}.
The optimal answer of FMRMC problem over different types of subgraphs has been provided in Section \ref{sec:Branches},
The optimal results for individual topologies with arbitrary and symmetric equilibrium distributions, have been presented in Section \ref{sec:FMMCPathBasedResults} and Section  \ref{sec:FMMCSymmetricResults}, respectively.
Section \ref{sec:Conclusions} concludes the paper.

\section{Fastest Mixing Reversible Markov Chain}
\label{sec:FMRMCsection}

We consider
an undirected
connected graph $\mathcal{G}=\left( \mathcal{V}, \mathcal{E} \right)$ consisting of the vertex set $\mathcal{V} = \{ 1, \ldots, N \}$ and the edge set $\mathcal{E} \subseteq \mathcal{V} \times \mathcal{V}$ where each edge $\{i,j\} \in \mathcal{E}$ is an ordered pair of distinct nodes with $\{i,j\} \in \mathcal{E} \Leftrightarrow \{j,i\} \in \mathcal{E}$.
The adjacency matrix of $\mathcal{G}$ can be written as below,
\begin{equation}
	\label{eq:Eq20171227333}
	\begin{gathered}
		\boldsymbol{A}_{\mathcal{G}} \left( i, j \right)
		=
		\begin{cases}
			1 &\text{if} \;\; \{i,j\} \in \mathcal{E},
			\\
			0 &\text{Otherwise.}
		\end{cases}
	\end{gathered}
\end{equation}

Markov chain over the graph can be defined by associating the state of Markov chain with vertices in the graph, i.e. $X(t) \in \mathcal{V}$, for $t=0,1,...$.
Correspondingly, transition of Markov chain between two adjacent states (vertices), is modelled by the edges in the graph, where each edge $\{i,j\}$ is associated with a transition probability $\boldsymbol{P}_{i,j}$.
For the Markov chain to be valid, the transition probability of each edge should be nonnegative (i.e. $\boldsymbol{P}_{i,j} \geq 0$) and the sum of the probabilities of outgoing edges of each vertex (including the self-loop edge) should be one.
A fundamental property of the irreducible Markov chains is the unique equilibrium distribution $(\pi)$, that satisfies
$\pi_{j}  =  \sum\nolimits_{i} \pi_{i} \boldsymbol{P}_{i,j}$ for all $j$ \cite{AldousBook2}.
The convergence theorem \cite{AldousBook2} states that for any initial distribution $P\left(X(t) = j \right) \rightarrow \pi_{j}$ as $t \rightarrow \infty$, for all $j$, provided the chain is aperiodic and irreducible.
The Markov chain is called reversible if the detailed balance equations hold \cite{AldousBook2}, i.e.
\begin{equation}
	\label{eq:Eq20171216220}
	\begin{gathered}
		\pi_{i} \boldsymbol{P}_{i,j}  =  \pi_{j} \boldsymbol{P}_{j,i} \;\; \text{for all} \;\; i,j.
	\end{gathered}
\end{equation}
The asymptotic convergence rate of the Markov chain to the equilibrium distribution $(\pi)$ is determined by the Second Largest Eigenvalue Modulus (SLEM) of the transition probability matrix \cite{BoydFastestmixing2003}, defined as
\begin{equation}
	\label{eq:Eq20171216227}
	\begin{gathered}
		\mu \left( \boldsymbol{P} \right)
		=
		\max_{i=2, ..., N} \left| \lambda_{i} \left( \boldsymbol{P} \right) \right|
		=
		\max\{ \lambda_{2} \left( \boldsymbol{P} \right), -\lambda_{N} \left( \boldsymbol{P} \right) \}.
	\end{gathered}
\end{equation}
where
$-1 < \lambda_N ( \boldsymbol{P} ) \leq \cdots \leq  \lambda_2 ( \boldsymbol{P} )  < \lambda_1 ( \boldsymbol{P} ) = 1$
are the eigenvalues of $\boldsymbol{P}$ in nonincreasing order,
which are real and no more than one in magnitude \cite{AldousBook2}.
Smaller SLEM results in faster converges to the equilibrium distribution.
The quantity $\log\left( 1 / \mu \left( \boldsymbol{P} \right)  \right)$ is referred to as the mixing rate, and $\tau = 1 / \log\left( 1 / \mu \left( \boldsymbol{P} \right)  \right)$ as the mixing time \cite{BoydFastestmixing2003}.
The FMMC problem for reversible Markov chains can be interpreted as the following optimization problem.
\begin{equation}
	\label{eq:Eq201712161073}
	\begin{aligned}
		\min\limits_{\boldsymbol{P}}
		\;\;
		&\mu\left(\boldsymbol{P}\right),
		\\
		s.t.
		\quad
		&
		\boldsymbol{P} \geq 0,
		\;\;
		\boldsymbol{P} \boldsymbol{1} = \boldsymbol{1},
		\;\;
		\boldsymbol{\pi}^{T} \boldsymbol{P} = \boldsymbol{\pi}^{T},
		\;\;
		\boldsymbol{D} \boldsymbol{P} = \boldsymbol{P}^{T} \boldsymbol{D},
		\\
		&
		\boldsymbol{P}_{i,j} = 0 \; \text{for} \; \{i,j\} \notin \mathcal{E}.
	\end{aligned}
\end{equation}
where $\boldsymbol{P} \geq 0$ is element-wise,
$\boldsymbol{D} = diag\left( \pi_{1}, \pi_{2}, ..., \pi_{N} \right)$,
$\boldsymbol{1}$ is the vector of all one
and
$\boldsymbol{\pi} = [ \pi_{1}, \pi_{2}, ..., \pi_{N} ]^{T}$.
Note that the only constraint on $\pi_{i}$ is that $\pi_{i} > 0$ for $i=1, ..., N$ and there is no constraint for their summation to be equal to one.

\subsubsection{Metropolis Transition Probabilities}
\label{sec:Metropolis}
Another non-optimal choice for transition probabilities is the Metropolis transition probability
\begin{equation}
	\label{eq:Eq20191015-metropolis-transition-probability}
	\begin{gathered}
		\boldsymbol{P}_{i,j} =
		\begin{cases}
			\frac{ 1 }{ d_{i} } \min\{ \frac{ \boldsymbol{\pi}_{j} d_{i} }{ \boldsymbol{\pi}_{i} d_{j} } , 1 \} 
			\quad\; \text{if} \quad (i,j) \in \mathcal{E},
			\\
			1 - \sum_{j \in \mathcal{N}(i)} \boldsymbol{P}_{i,j} 
			\;\; \text{if} \quad i=j,
			\\
			0 \qquad\qquad\qquad\quad \;\; \text{otherwise},
		\end{cases}
	\end{gathered}
\end{equation}
where $\mathcal{N}(i)$ is the set of neighbors of $i$-th vertex.
The chain is guaranteed to be aperiodic if there is at least one vertex $i \in \mathcal{V}$ such that $\frac{  \boldsymbol{\pi}_{j} d_{i}  }{  \boldsymbol{\pi}_{i} d_{j}  } < 1 (\Rightarrow \boldsymbol{P}_{i,i} > 0)$ i.e. it is not the case that for all $i,j \in \mathcal{V}$: 
$ \frac{ \boldsymbol{\pi}_{i} }{ d_{i} } = \frac{ \boldsymbol{\pi}_{j} }{ d_{j} } =$ const.
This is guaranteed by setting the equilibrium distribution to $\boldsymbol{\pi}_{i} = d_{i} + 1$ for $i=1,...,N$.

\subsection{FMMC in terms of Symmetric Laplacian}
\label{sec:FMRMCSymmetricLaplacian}
In
\cite[Subsection~2.3]{JAFARIZADEH201913},
the FMRMC problem has been formulated as an optimization problem in terms of the symmetric Laplacian.
Following the same procedure, the FMMC problem
(\ref{eq:Eq201712161073})
can be formulated as the following optimization problem,
\begin{equation}
	\label{eq:Eq201712231257}
	\begin{aligned}
		\min\limits_{q_{i,j}}
		\;\;
		&
		\mu\left( \boldsymbol{I} - \boldsymbol{D}^{-1} \boldsymbol{L}\left( q \right) \right),
		\\
		s.t.
		\quad
		&
		q_{i,j} = q_{j,i} \geq 0,
		\;\;
		q_{i,j} = 0 \; \text{for} \; \{i,j\} \notin \mathcal{E},
		\\
		&
		- \sum\nolimits_{k \neq i} q_{ik} + \pi_{i} \geq 0 \;\; \text{for} \;\; i=1, ..., N.
	\end{aligned}
\end{equation}
where
$q_{i,j}$ are the weights defined as
$q_{i,j} = \pi_{i} \boldsymbol{P}_{i,j} = \pi_{j} \boldsymbol{P}_{j,i} = q_{j,i}$,
and
$\boldsymbol{L}(q) = \sum\nolimits_{ \{i,j\} \in \mathcal{E} } q_{ij} \left( \boldsymbol{e}_{i} - \boldsymbol{e}_{j} \right) \left( \boldsymbol{e}_{i} - \boldsymbol{e}_{j} \right)^{T}$
is the symmetric Laplacian matrix.
$\boldsymbol{e}_{i}$ for $i \in \mathcal{V}$ is a column vector with $i$-th element equal to $1$ and zero elsewhere.
Using the derivation in
(\ref{eq:Eq201712231257}),
the constraints $\boldsymbol{P} \boldsymbol{1} = \boldsymbol{1}$ and $\boldsymbol{\pi}^{T} \boldsymbol{P} = \boldsymbol{\pi}^{T}$ in
(\ref{eq:Eq201712161073}) 
are automatically satisfied.

\section{Fastest Mixing Reversible Markov Chain on Clique Lift of a Graph}
\label{sec:FMRMCCliqueLiftGraphs}

In this section, we address the FMRMC problem over clique lifted graphs.

\begin{definition}
	\label{LiftDefinition}
	Given
	graph $\mathcal{G}=\{ \mathcal{V}, \mathcal{E} \}$,
	graph $\widetilde{\mathcal{G}} = \{ \widetilde{\mathcal{V}}, \widetilde{\mathcal{E}} \}$
	is the lift of $\mathcal{G}$ if there exists a surjection $f: \widetilde{\mathcal{V}} \rightarrow \mathcal{V}$, such that
	$(i, j) \in \widetilde{\mathcal{E}}$
	implies that
	$( f(i), f(j) ) \in \mathcal{E}$.
	The graph $\mathcal{G}$ is referred to as the base graph and for every vertex $v \in \mathcal{V}$, the $f^{-1} \left( v \right)$ is called the fiber of $v$ \cite{Godsil2004}.
\end{definition}

\subsection{Clique Lift of a Graph}

Clique lift of a given graph $\mathcal{G}$ is obtained by replacing every node in $\mathcal{G}$ by a clique of size $m_{i} > 0$, such that if $\{i,j\} \in \mathcal{E}$ then all nodes in their corresponding cliques are connected to each other as well.
Clique Lift of a graph $\mathcal{G} = \left( \mathcal{V}, \mathcal{E} \right)$ is denoted by $\widetilde{\mathcal{G}}=\left( \widetilde{\mathcal{V}}, \widetilde{\mathcal{E}} \right)$,
where $\mathcal{G}$ is the base graph and
the fiber of $i$-th vertex in $\mathcal{G}$ is a complete graph (or clique) of size $m_{i}$.
Three different clique lift graphs of a path graph with three vertices are depicted in Fig. \ref{fig:LiftPath3}.
The adjacency matrix of the base graph $\boldsymbol{A}_{\mathcal{G}}$, given in (\ref{eq:Eq20171227333}), can also be written as
$\boldsymbol{A}_{\mathcal{G}}  =  \sum\nolimits_{i \in \mathcal{V}} \sum\nolimits_{j \in \mathcal{N}_{i}} \boldsymbol{E}_{i,j}  = \sum\nolimits_{\{i,j\} \in \mathcal{E}} \boldsymbol{E}_{i,j}$,
where $\mathcal{N}_{i}$ is the set of neighbors of node $i$ and $\boldsymbol{E}_{i,j}$ is a $|\mathcal{V}| \times |\mathcal{V}|$ matrix with element on $i$-th row and $j$-th column equal to $1$ and zero elsewhere.
The adjacency matrix of $\widetilde{\mathcal{G}}$ is a $\sum_{i=1}^{|\mathcal{V}|} m_{i} \times \sum_{i=1}^{|\mathcal{V}|} m_{i}$ matrix with $|\mathcal{V}|$ partitions, where the $(i,j)$-th partition is of size $m_{i} \times m_{j}$ for $i, j = 1, ..., |\mathcal{V}|$.
The adjacency matrix $\boldsymbol{A}_{\widetilde{\mathcal{G}}}$ can be written as below
\begin{subequations}
	\label{eq:Eq20180102493}
	\begin{gather}
		\boldsymbol{A}_{\widetilde{\mathcal{G}}} \left[ i, i \right]
		=
		\boldsymbol{J}_{m_{i}, m_{i}} - \boldsymbol{I}_{m_{i}},
		\label{eq:Eq20180102493a}
		\\
		\boldsymbol{A}_{\widetilde{\mathcal{G}}} \left[ i, j \right]
		=
		\begin{cases}
			\boldsymbol{J}_{m_{i}, m_{j}} \;\; \text{for} \;\; \{i, j\} \in \mathcal{E}
			\\
			\boldsymbol{0} \;\; \text{otherwise}
		\end{cases}
		\label{eq:Eq20180102493b}
	\end{gather}
\end{subequations}
where
in (\ref{eq:Eq20180102493a}), $\boldsymbol{A}_{\widetilde{\mathcal{G}}} \left[ i, i \right]$ for $i = 1, ..., |\mathcal{V}|$ denotes the diagonal partitions
and
in (\ref{eq:Eq20180102493b}), $\boldsymbol{A}_{\widetilde{\mathcal{G}}} \left[ i, j \right]$ for $i \neq j = 1, ..., |\mathcal{V}|$, denotes the off-diagonal partitions.
$\boldsymbol{J}_{m_{i}, m_{j}}$ is a $m_{i} \times m_{j}$ matrix with all elements equal to one,
$\boldsymbol{I}_{m_{i}}$ is the identity matrix of size $m_{i}$,
and $m_{i}$ is the size of clique replacing $i$-th node in $\mathcal{G}$.
In the clique lift graph $\widetilde{\mathcal{G}}$, we denote each vertex by $(i, \alpha)$,
which refers to $\alpha$-th vertex within $i$-th partition (or clique), for $i = 1, ..., |\mathcal{V}|$ and $\alpha = 1, ..., m_{i}$.

\subsection{FMRMC Problem on Clique Lift of a Graph}
\label{sec:FMRMCCliqueLiftGraphs-MainTheorem}

\begin{theorem}
	\label{CauchyInterlacingTheorem}
	Given the base graph $\mathcal{G}$
	(with equilibrium distribution $\pi_{i}$)
	and its clique lifted graph $\widetilde{\mathcal{G}}$
	(with equilibrium distribution $\widetilde{\pi}_{i, \alpha}$),
	the FMRMC problem
	(with any given equilibrium distribution)
	over the clique lifted graph $\widetilde{\mathcal{G}}$
	is reducible to
	the FMRMC problem over the base graph $\mathcal{G}$
	where
	the optimal $SLEM$ of the clique lifted graph $\widetilde{\mathcal{G}}$ is equal to that of the base graph $\mathcal{G}$ and
	the following relations hold between their corresponding optimal results,
	\begin{subequations}
		\label{eq:Eq201801061374}
		\begin{gather}
			\widetilde{q}(i,j)_{\alpha, \beta}
			=
			\left( q_{i,j} / \left( \pi_{i} \pi_{j} \right) \right)
			\widetilde{\pi}_{i,\alpha} \widetilde{\pi}_{j, \beta}
			\label{eq:Eq201801061374a}
			\\
			\widetilde{q}(i)_{\alpha_i, \beta_i}
			=
			( (\pi_i-\sum\nolimits_{j\in \mathcal{N}(i)} q_{i,j})\widetilde{\pi}_{i, \alpha_i} \widetilde{\pi}_{i, \beta_i} ) / \pi_i^2
			\label{eq:Eq201801061374b}
		\end{gather}
	\end{subequations}
	and
	the following
	holds between the corresponding equilibrium distributions of the base and lifted graphs.
	\begin{equation}
		\label{eq:Eq20180101663}
		\begin{gathered}
			\pi_i = \sum\nolimits_{\alpha=1}^{m_i} \widetilde{\pi}_{i, \alpha}
			\;\;
			\text{for}
			\;\;
			i \in \mathcal{V},
		\end{gathered}
	\end{equation}
	where (\ref{eq:Eq201801061374a}) holds for
	for all $\{i,j\} \in \mathcal{E}$
	and $\alpha=1, ..., m_i$, $\beta = 1, ..., m_j$
	and
	(\ref{eq:Eq201801061374b}) holds for $i \in \mathcal{V}$.
\end{theorem}

\begin{proof}
	Considering $\widetilde{\pi}_{i, \alpha}$ as the equilibrium distribution of the $(i, \alpha)$-th vertex in the clique lift graph, the diagonal matrix $\widetilde{\boldsymbol{D}}$ is defined as
	$\widetilde{\boldsymbol{D}} = diag\left( \widetilde{\boldsymbol{D}}_{1}, \widetilde{\boldsymbol{D}}_{2}, ..., \widetilde{\boldsymbol{D}}_{|\mathcal{V}|} \right)$,
	where
	$\widetilde{\boldsymbol{D}}_{i} = diag\left( \widetilde{\pi}_{i, 1}, \widetilde{\pi}_{i, 2}, ..., \widetilde{\pi}_{i, m_{i}} \right)$
	is the diagonal matrix representing $i$-th partition (or clique).
	The symmetric Laplacian of the clique lift graph $\widetilde{\mathcal{G}}$ is a $\sum_{i=1}^{|\mathcal{V}|} m_{i} \times \sum_{i=1}^{|\mathcal{V}|} m_{i}$ matrix with $|\mathcal{V}|$ partitions, where the $(i,j)$-th partition is of size $m_{i} \times m_{j}$ for $i, j = 1, ..., |\mathcal{V}|$.
	The symmetric Laplacian matrix $\boldsymbol{L} \left( \widetilde{q} \right)$ can be written as below
	\begin{subequations}
		\label{eq:Eq20180103670}
		\begin{gather}
				\boldsymbol{L} \left( \widetilde{q} \right) \left[ i, i \right] \left( \alpha, \alpha \right)
				=
				\sum\nolimits_{\substack{\zeta=1 \\ \zeta \neq \alpha}}^{m_{i}} \widetilde{q}\left( i \right)_{\alpha, \zeta}
				+
				\sum\nolimits_{j \in \mathcal{N}(i)} \sum\nolimits_{\zeta=1}^{m_{j}} \widetilde{q}\left( i, j \right)_{\alpha, \zeta},
			\label{eq:Eq20180103670a}
			\\
			\boldsymbol{L} \left( \widetilde{q} \right) \left[ i, i \right] \left( \alpha, \beta \right)
			=
			- \widetilde{q}\left( i \right)_{\alpha, \beta},
			\label{eq:Eq20180103670b}
			\\
			\boldsymbol{L} \left( \widetilde{q} \right) \left[ i, j \right] \left( \alpha, \beta \right)
			=
			- \widetilde{q}\left( i, j \right)_{\alpha, \beta},
			\label{eq:Eq20180103670c}
		\end{gather}
	\end{subequations}
	where
	$\widetilde{q}\left( i, j \right)_{\alpha, \beta}$ is the corresponding value of $q$ in the clique lifted graph between $\alpha$-th and $\beta$-th vertices in the $i$-th and $j$-th cliques, respectively, which in turn correspond to $i$-th and $j$-th nodes in the base graph.
	Similarly,
	$\widetilde{q}\left( i \right)_{\alpha, \beta}$
	is the corresponding value of $q$ in the clique lifted graph between $\alpha$-th and $\beta$-th vertices in the $i$-th clique.
	In (\ref{eq:Eq20180103670a}) and (\ref{eq:Eq20180103670b}), $\alpha \neq \beta = 1, ..., m_{i}$ and $\boldsymbol{L} \left( \widetilde{q} \right) \left[ i, i \right]$ for $i \in \mathcal{V}$ denotes the diagonal partitions.
	In (\ref{eq:Eq20180103670c}), $\boldsymbol{L} \left( \widetilde{q} \right) \left[ i, j \right]$ for $\{i, j\} \in \mathcal{E}$ denotes the off-diagonal partitions and $\alpha = 1, ..., m_{i}$,  and  $\beta = 1, ..., m_{j}$.
	The transition probability matrix $( \widetilde{\boldsymbol{P}} )=\boldsymbol{I} - \widetilde{\boldsymbol{D}}^{-1}\boldsymbol{L}( \widetilde{q} )$ of the clique lift graph is as below,
	\begin{subequations}
		\label{eq:Eq20180103771}
		\begin{gather}
				\widetilde{\boldsymbol{P}} \left[ i, i \right] \left( \alpha, \alpha \right)
				=1
				-\sum\nolimits_{\substack{\zeta=1 \\ \zeta \neq \alpha}}^{m_{i}}\frac{ \widetilde{q}\left( i \right)_{\alpha, \zeta}}{\widetilde{\pi}_{i,\alpha}}
				-
				\sum\nolimits_{j \in \mathcal{N}(i)} \sum\nolimits_{\zeta=1}^{m_{j}} \frac{\widetilde{q}\left( i, j \right)_{\alpha, \zeta}}{\widetilde{\pi}_{i,\alpha}}
			\label{eq:Eq20180103771a}
			\\
			\widetilde{\boldsymbol{P}}  \left[ i, i \right] \left( \alpha, \beta \right)
			=
			\widetilde{q}\left( i \right)_{\alpha, \beta} / \widetilde{\pi}_{i,\alpha}
			\label{eq:Eq20180103771b}
			\\
			\widetilde{\boldsymbol{P}} \left[ i, j \right] \left( \alpha, \beta \right)
			=
			\widetilde{q}\left( i, j \right)_{\alpha, \beta} / \widetilde{\pi}_{i,\alpha}
			\label{eq:Eq20180103771c}
		\end{gather}
	\end{subequations}
	In (\ref{eq:Eq20180103771a}) and (\ref{eq:Eq20180103771b}), $\alpha \neq \beta = 1, ..., m_{i}$ and $\widetilde{\boldsymbol{P}}  \left[ i, i \right]$ for $i \in \mathcal{V}$ denotes the diagonal partitions.
	In (\ref{eq:Eq20180103771c}), $\widetilde{\boldsymbol{P}}  \left[ i, j \right]$ for $\{i, j\} \in \mathcal{E}$ denotes the off-diagonal partitions and $\alpha = 1, ..., m_{i}$,  and  $\beta = 1, ..., m_{j}$.
	Defining the projection operator $\Lambda$
	\begin{equation}
		\boldsymbol{\Lambda} \left[ i,j \right]
		=
		\begin{cases}
			\frac{\boldsymbol{J}_{1, m_{i}}}{\sqrt{m_{i}}} \;\; \text{for} \;\; i= j
			\\
			\boldsymbol{0} \;\; \;\; \text{for} \;\; i\neq j
		\end{cases}
	\end{equation}
	which satisfies $\Lambda \times \Lambda^T=I_{|\mathcal{V}|}$, for the transition probability matrix $\widetilde{\boldsymbol{P}}$ we have
	\begin{equation}
		\label{eq:Eq20180103877}
		\begin{gathered}
			\boldsymbol{P}^{'} = \boldsymbol{\Lambda} \times \widetilde{\boldsymbol{P}} \times \boldsymbol{\Lambda}^{T}
		\end{gathered}
	\end{equation}
	Using the similarity transformation $\boldsymbol{S} = \sum\nolimits_{i=1}^{|\mathcal{V}|} \boldsymbol{E}_{i,i} \frac{1}{\sqrt{m_i}}$, we have
	$\boldsymbol{P} = \boldsymbol{S} \times \boldsymbol{P}^{'} \times \boldsymbol{S}^{-1}$
	where $\boldsymbol{P}$ is as below
	\begin{equation}
		\label{eq:Eq20171230951a}
		\begin{gathered}
				\boldsymbol{P}
				=
				\sum\nolimits_{i=1}^{|\mathcal{V}|} \left( 1 - \sum\nolimits_{j\in \mathcal{N}(i)} \frac{q_{i,j}}{\pi_i} \right) \boldsymbol{E}_{i,i}
				+
				\sum\nolimits_{i=1}^{|\mathcal{V}|} \sum\nolimits_{j \in \mathcal{N}(i)} \left( \frac{q_{i,j}}{\pi_i} \right)\boldsymbol{E}_{i,j}
		\end{gathered}
	\end{equation}
	with
	(\ref{eq:Eq20180101663})
	and
	\begin{equation}
		\label{eq:Eq20180101678ww}
		\begin{gathered}
			( q_{i,j} / \pi_i )
			=
			\left( \sum\nolimits_{ \alpha = 1 }^{ m_{i} } \sum\nolimits_{ \beta = 1 }^{ m_{j} } ( \widetilde{q}\left( i, j \right)_{\alpha, \beta}  /  \widetilde{\pi}_{i, \alpha} ) \right)
			/ m_{i}
		\end{gathered}
	\end{equation}
	and $\boldsymbol{E}_{i,j}$ is a $|\mathcal{V}| \times |\mathcal{V}|$ matrix with the element on $i$-th row and $j$-th column equal to one and zero elsewhere.
	Based on (\ref{eq:Eq20180103877}) and Cauchy interlacing theorem (Appendix \ref{sec:CauchyInterlacingTheorem}), it can be concluded that
	the eigenvalues of $\boldsymbol{P}^{'}$ interlace the eigenvalues of $\widetilde{\boldsymbol{P}}$.
	From this it can be concluded that
	the smallest (second largest) eigenvalue of $\widetilde{\boldsymbol{P}}$ given in
	(\ref{eq:Eq20180103771})
	is smaller (larger) than
	the smallest (second largest) 
	eigenvalue of $\boldsymbol{P}^{'}$.
	Thus the SLEM of $\widetilde{\boldsymbol{P}}$ is greater than that of $\boldsymbol{P}^{'}$.
	On the otherhand, since $\boldsymbol{S}$ is a similarity transformation then $\boldsymbol{P}$ and $\boldsymbol{P}^{'}$ have the same spectrum and it can be concluded that
	the SLEM of $\widetilde{\boldsymbol{P}}$ is greater than that of $\boldsymbol{P}$.
	The main question here is if there is a $\widetilde{q}(i,j)_{\alpha, \beta}$ for given $q_{i,j}$ such that the spectrum of (\ref{eq:Eq20180103771}) and (\ref{eq:Eq20171230951a}) are equal, while (\ref{eq:Eq20180101678ww}) is satisfied.
	In the following we show that such $\widetilde{q}(i,j)_{\alpha, \beta}$ is achievable.
	By appropriate choice of $\widetilde{q}(i,j)_{\alpha_i ,\beta_j}$ for each $\{i,j\} \in \mathcal{E}$ and $\alpha_i = 1, ..., m_i$, $\beta_j = 1, ..., m_j$,
	we can reduce
	the matrix
	$\widetilde{\boldsymbol{P}} \left[ i, j \right]$
	into a rank one matrix.
	To this aim, it is necessary for the following to hold true
	\begin{equation}
		\label{eq:Eq20171230630}
		\begin{gathered}
			\frac{ \widetilde{q}(i,j)_{\alpha_i,1} }  { \widetilde{\pi}_{i,\alpha_i} \widetilde{\pi}_{j,1} }
			=
			\frac{ \widetilde{q}(i,j)_{\alpha_i,2} }  { \widetilde{\pi}_{i,\alpha_i} \widetilde{\pi}_{j,2} }
			=
			\cdots
			=
			\frac{ \widetilde{q}(i,j)_{\alpha_i,m_j} }  { \widetilde{\pi}_{i,\alpha_i} \widetilde{\pi}_{j,m_j} }
		\end{gathered}
	\end{equation}
	for all $\{i,j\} \in \mathcal{E}$ and $\alpha_i=1, \cdots, m_i$,
	where (\ref{eq:Eq201801061374a}) can be derived from (\ref{eq:Eq20171230630}).
	Choosing $\widetilde{q}(i,j)_{\alpha, \beta}$ according to
	(\ref{eq:Eq201801061374a}),
	guarantees that (\ref{eq:Eq20180101678ww}) is satisfied.

	Similarly by appropriate choice of
	$\widetilde{q}(i)_{\alpha_i ,\beta_i }$, for each $i \in \mathcal{V}$
	and
	$\alpha_i \neq \beta_i = 1, ..., m_i$, we can reduce
	the matrix
	$\widetilde{\boldsymbol{P}}\left[ i, i \right]$
	into a rank one matrix.
	To this aim, it is necessary for the following to hold true
	\begin{equation}
		\label{eq:Eq20171230688}
		\begin{gathered}
			\frac{ \widetilde{q}(i)_{\alpha, 1}}{ \widetilde{\pi}_{i, \alpha} \widetilde{\pi}_{i,1} }
			=
			\frac{ \widetilde{q}(i)_{\alpha, 2 }}{ \widetilde{\pi}_{i, \alpha} \widetilde{\pi}_{i,2} }
			=
			\cdots
			=
			\frac{ \widetilde{q}(i)_{\alpha, m_i}}{ \widetilde{\pi}_{i,\alpha} \widetilde{\pi}_{i,m_i} }
		\end{gathered}
	\end{equation}
	for all $i \in \mathcal{V}$ and $\alpha_i=1, ..., m_i$.
	From (\ref{eq:Eq20171230688}), we have
	\begin{equation}
		\label{eq:Eq20171230712}
		\begin{gathered}
			\widetilde{q}(i)_{\alpha_i, \beta_i}
			=
			Q_{i} \widetilde{\pi}_{i, \alpha_i} \widetilde{\pi}_{i, \beta_i}
		\end{gathered}
	\end{equation}
	for $i \in \mathcal{V}$ and $\alpha_i=1, ..., m_i$.
	In the following, the coefficients $Q_{i}$ will be determined.
	Defining
	the unitary DFT matrix of size $m_i \times m_i$, for $i=1, ...,|\mathcal{V}|$ in each fiber, i.,e.
	\begin{equation}
		\label{eq:Eq201808282217}
		\begin{gathered}
			\boldsymbol{F} \left[ i,j \right]
			=
			\begin{cases}
				\boldsymbol{F}_i \;\; \text{for} \;\; i= j
				\\
				\boldsymbol{0} \;\; \;\; \text{for} \;\; i\neq j
			\end{cases}
		\end{gathered}
	\end{equation}
	with $\boldsymbol{F}_{i}$
	as a square matrix with $(\alpha,\beta)$-th element equal to
	$\boldsymbol{F}_{i} \left( \alpha, \beta \right) = \frac{1}{\sqrt{m_{i}}} \omega_{i}^{ ( \alpha - 1 ) ( \beta - 1 ) }$,
	for $\alpha, \beta = 1, ..., m_{i}$,
	and
	$\omega_{i} = e^{j \frac{2\pi}{m_{i}}}$,
	where
	$\boldsymbol{F}_{i} \boldsymbol{F}_{i}^{\dag}  =  \boldsymbol{F}_{i}^{\dag} \boldsymbol{F}_{i} = \boldsymbol{I}_{m_{i}}$
	and
	$\boldsymbol{F}_{i}^{\dag}$ is the conjugate transpose of $\boldsymbol{F}_{i}$.
	Considering the following similarity transformation in each fiber,
	\begin{equation}
		\begin{gathered}
			\boldsymbol{\tilde{S}} \left[ i,j \right]
			=
			\begin{cases}
				\frac{1}{\sqrt{m_i}}\boldsymbol{I}_{m_i} \;\; \text{for} \;\; i= j
				\\
				\boldsymbol{0} \;\; \text{for} \;\; i\neq j
			\end{cases}
		\end{gathered}
	\end{equation}
	and
	by choosing $\widetilde{q}(i, j)_{\alpha, \beta}$ and $\widetilde{q}(i)_{\alpha_i, \beta_i}$ according to
	(\ref{eq:Eq201801061374a})
	and (\ref{eq:Eq20171230688}),
	the matrix $\widetilde{\boldsymbol{P}}\left[ i, j \right]$ reduces to a rank one matrix, as below
	\begin{equation}
		\label{eq:201801061233}
		\begin{aligned}
			&
			\left(\boldsymbol{\widetilde{S}} \times \boldsymbol{F} \times \widetilde{\boldsymbol{P}} \times \boldsymbol{F}^\dag \times \boldsymbol{\widetilde{S}^{-1}} \right) \left[ i, j \right] =
			( q_{i,j} / ( \pi_{i} \pi_{j} ) )
			\times
			\\&
			\qquad
				\left[ \begin{array}{ccccc}
					{\pi_{j}}
					&{\sum\limits_{\beta=1}^{m_j} \omega^{\beta-1}_{j} \widetilde{\pi}_{j,\beta}}
					&{\sum\limits_{\beta=1}^{m_j}\omega^{2(\beta-1)}_{j} \widetilde{\pi}_{j,\beta}}
					&{\cdots}
					&{\sum\limits_{\beta=1}^{m_j}\omega^{(m_j-1)(\beta-1)}_{j} \widetilde{\pi}_{j,\beta}}
					\\
					{0}
					&{0}
					&{0}
					&{\cdots}
					&{0}
					\\
					{\vdots}
					&{\vdots}
					&{\vdots}
					&{\ddots}
					&{\vdots}
					\\
					{0}
					&{0}
					&{0}
					&{\cdots}
					&{0}
				\end{array} \right],
		\end{aligned}
	\end{equation}
	\begin{equation}
		\label{eq:201801061288}
		\begin{aligned}
			&
			\left( \boldsymbol{\widetilde{S}} \times \boldsymbol{F} \times \widetilde{\boldsymbol{P}} \times \boldsymbol{F}^\dag \times \boldsymbol{\widetilde{S}^{-1}} \right) \left[ i, i \right] =
			\\
			&
				\left[ \begin{array}{ccccc}
					{1-\sum_{j\in \mathcal{N}(i)}\frac{ q_{i,j} }{ \pi_{i} }}
					&{Q_i(\sum_{\alpha=1}^{m_i}\omega^{\alpha-1}_{i} \widetilde{\pi}_{i,\alpha})}
					&{Q_i(\sum_{\alpha=1}^{m_i}\omega^{2(\alpha-1)}_{i} \widetilde{\pi}_{i,\alpha})}
					&{\cdots}
					&{Q_i(\sum_{\alpha=1}^{m_i}\omega^{(n_i-1)(\alpha-1)}_{i} \widetilde{\pi}_{i,\alpha})}
					\\
					{0}
					&{s_{i}}
					&{0}
					&{0}
					\\
					{0}
					&{0}
					&{s_{i}}
					&{\cdots}
					&{0}
					\\
					{\vdots}
					&{\vdots}
					&{\vdots}
					&{\cdots}
					&{\vdots}
					\\
					{0}
					&{0}
					&{0}
					&{\cdots}
					&{s_{i}}
				\end{array} \right]
		\end{aligned}
	\end{equation}
	where
	$s_{i} = 1-\sum_{j \in \mathcal{N}(i) } \frac{ q_{i,j} }{ \pi_{i} \pi_{j} } \left( \sum_{\alpha=1}^{m_j} \widetilde{\pi}_{j,\alpha} \right) - Q_i \left( \sum_{\alpha=1}^{m_i} \widetilde{\pi}_{i,\alpha} \right)$
	for $i \in \mathcal{V}$.
	$s_{i}$ for $i \in \mathcal{V}$ are the eigenvalues of $\widetilde{\boldsymbol{P}}$.
	Setting $s_{i} = 0$ reduces above matrices to rank one matrices.
	From $s_{i} = 0$,
	the
	following can be concluded
	\begin{equation}
		\label{eq:Eq201809072588}
		\begin{gathered}
				Q_{i}
				=
				\frac
				{ 1 - \sum_{j\in \mathcal{N}(i)}\frac{ q_{i,j} }{ \pi_{i} \pi_{j} } \left( \sum_{\alpha=1}^{m_j} \widetilde{\pi}_{j, \alpha} \right) }
				{ \sum_{\alpha=1}^{m_i} \widetilde{\pi}_{i,\alpha} }
				=
				\frac{\pi_i-\sum_{j\in \mathcal{N}(i)} q_{i,j} }{\pi_i^2}
		\end{gathered}
	\end{equation}
	for $i \in \mathcal{V}$,
	where considering (\ref{eq:Eq201809072588}) from (\ref{eq:Eq20171230712}), we can derive (\ref{eq:Eq201801061374b}).
	$\boldsymbol{\widetilde{S}} \times \boldsymbol{F} \times \widetilde{\boldsymbol{P}} \times \boldsymbol{F}^\dag \times \boldsymbol{\widetilde{S}}^{-1} $ is the similarity transform of $\widetilde{\boldsymbol{P}} $.
	Thus considering (\ref{eq:201801061233}) and (\ref{eq:201801061288}) (with $s_{i} = 0$ for $i \in \mathcal{V}$) based on the particular choice of $\widetilde{q}(i,j)_{\alpha, \beta}$  and $\widetilde{q}(i)_{\alpha, \beta}$ given in
	(\ref{eq:Eq201801061374}),
	we see that the nonzero eigenvalues of $\widetilde{\boldsymbol{P}}$ are the same as those of $\boldsymbol{P}$ given in (\ref{eq:Eq20171230951a}), hence, their corresponding SLEM are the same,
	and Theorem \ref{CauchyInterlacingTheorem} can be deduced.

\end{proof}

\begin{remark}
	\label{remark:TightInterlacing}
	Considering the base graph $\mathcal{G}$, and one of its clique lifted graphs $\widetilde{\mathcal{G}}$, based on Theorem \ref{CauchyInterlacingTheorem}, it can be stated that tight interlacing (as defined in \cite{HAEMERS1995593}) with $k$ at least equal to $1$  holds between
	the optimal eigenvalues of their corresponding transition probability matrices.
\end{remark}

\begin{remark}
	\label{remark:Interpretation-of-24}
	The interpretation of
	(\ref{eq:Eq201801061374a})
	is that
	in the clique lifted graph
	the	transition probabilities from vertices
	$\{(j,1), ..., (j,m_{j})\}$ to the vertex $(i,\alpha_{i})$
	for all $\{ i, j \} \in \mathcal{E}$ and $\alpha_{i} = 1, ..., m_{i}$
	are all equal, i.e.,
	$ \widetilde{\boldsymbol{P}} \left[ j, i \right] \left( 1, \alpha_{i} \right) 		= 		... 		= 		\widetilde{\boldsymbol{P}} \left[ j, i \right] \left( m_{j}, \alpha_{i} \right) $.
	Also, The interpretation of
	(\ref{eq:Eq201801061374b})
	is that
	inside every clique in the clique lifted graph
	the	transition probabilities from all other vertices of the clique to any particular vertex of the same clique
	are equal, i.e.,
	$ \widetilde{\boldsymbol{P}} \left[ i, i \right] \left( 1, \alpha_{i} \right) 		= 		... 		= 		\widetilde{\boldsymbol{P}} \left[ i, i \right] \left( m_{j}, \alpha_{i} \right) $,
	for $\alpha_{i} = 1, ..., m_{i}$.
\end{remark}

\begin{figure}
	\centering
	\includegraphics[width=120mm]{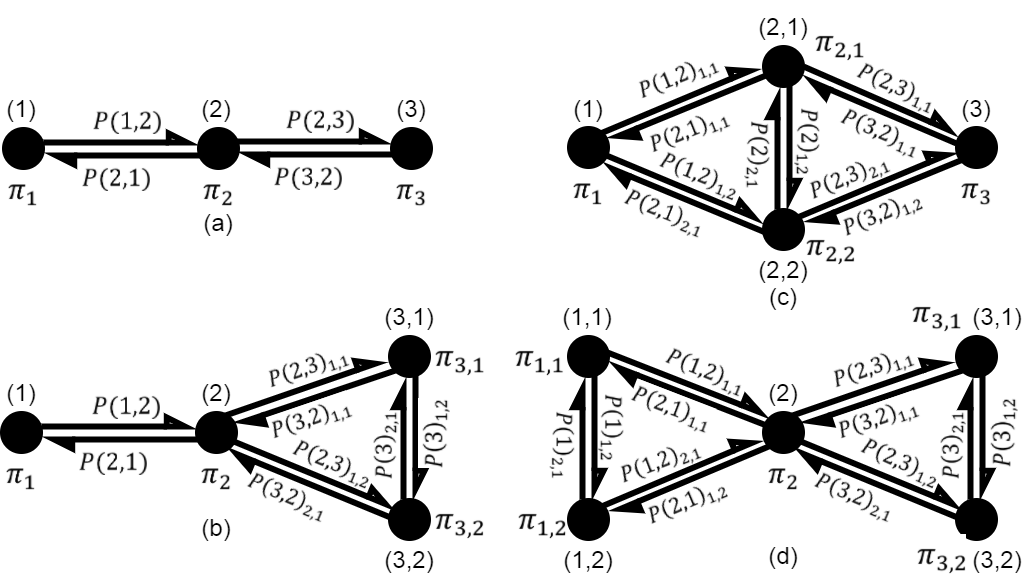}
	\caption{(a) Path graph with $3$ vertices as the base graph, (b) Lifted graph with $m_1 = 1$, $m_2 = 1$, $m_3 = 2$, (c) Lifted graph with $m_1 = 1$, $m_2 = 2$, $m_3 = 1$, (d) Lifted graph with $m_1 = 2$, $m_2 = 1$, $m_3 = 2$.}
	\label{fig:LiftPath3}
\end{figure}

\section{Semidefinite Programming Formulation}
\label{sec:FMRMCCliqueLiftGraphsSDP}

Considering the FMRMC problem over the base graph,
in this section, we formulate the optimization problem (\ref{eq:Eq201712231257}) in the form of semidefinite programming problem.

Following a procedure similar to that of \cite{BoydFastestmixing2003}, it can be shown that FMMC problem
(\ref{eq:Eq201712231257})
is a convex optimization problem and can be derived as the following semidefinite programming problem,
\begin{equation}
	\label{eq:Eq201711251293}
	\begin{aligned}
		\min
		\;\;
		&s,
		\\
		s.t.
		\quad
		&
		-s \boldsymbol{I}
		\preccurlyeq
		\boldsymbol{I} - \boldsymbol{D}^{-\frac{1}{2}} \boldsymbol{L} \left( q \right) \boldsymbol{D}^{-\frac{1}{2}} -
		\widetilde{\boldsymbol{J}}
		\preccurlyeq
		s \boldsymbol{I},
		\\
		&
		- \sum\nolimits_{k \neq i} q_{ik} + \pi_{i} \geq 0 \;\; \text{for} \;\; i=1, ..., N.
	\end{aligned}
\end{equation}
The symbol $\preccurlyeq$ denotes matrix inequality, i.e., $\boldsymbol{X} \preccurlyeq \boldsymbol{Y}$ means $\boldsymbol{Y} - \boldsymbol{X}$ is positive semidefinite.
In optimization problem (\ref{eq:Eq201711251293})
the optimization variables are $q_{ij}$.
Considering
\begin{equation}
	\label{eq:Eq20171112317}
	\begin{gathered}
			\boldsymbol{D}^{-\frac{1}{2}} \boldsymbol{L} \left( q \right) \boldsymbol{D}^{-\frac{1}{2}}
			=
			\sum_{\{i,j\} \in \mathcal{E}} q_{ij} \boldsymbol{D}^{-\frac{1}{2}} \left( \boldsymbol{e}_{i} -\boldsymbol{e}_{j} \right)  \left( \boldsymbol{e}_{i} -\boldsymbol{e}_{j} \right)^{T}  \boldsymbol{D}^{-\frac{1}{2}}
	\end{gathered}
\end{equation}
problem (\ref{eq:Eq201711251293}) can be written as below,
\begin{subequations}
	\label{eq:Eq201801262501}
	\begin{align}
		\min
		\;\;
		&s,
		\nonumber
		\\
		s.t.
		\quad
		&
			(s - 1) \boldsymbol{I}
			+
			\sum\nolimits_{\{i,j\} \in \mathcal{E}} q_{ij} \boldsymbol{D}^{-\frac{1}{2}} \left( \boldsymbol{e}_{i} -\boldsymbol{e}_{j} \right)  \left( \boldsymbol{e}_{i} -\boldsymbol{e}_{j} \right)^{T}  \boldsymbol{D}^{-\frac{1}{2}}
			+
			\widetilde{\boldsymbol{J}}
			\geq 0
		\label{eq:Eq201801262501a}
		\\
		&
			(s + 1) \boldsymbol{I}
			-
			\sum\nolimits_{\{i,j\} \in \mathcal{E}} q_{ij} \boldsymbol{D}^{-\frac{1}{2}} \left( \boldsymbol{e}_{i} -\boldsymbol{e}_{j} \right)  \left( \boldsymbol{e}_{i} -\boldsymbol{e}_{j} \right)^{T}  \boldsymbol{D}^{-\frac{1}{2}}
			-
			\widetilde{\boldsymbol{J}}
			\geq 0
		\label{eq:Eq201801262501b}
		\\
		&
			- \sum\nolimits_{k \neq i} q_{ik} + \pi_{i} \geq 0 \;\; \text{for} \;\; i=1, ..., N.
		\label{eq:Eq201801262501c}
	\end{align}
\end{subequations}
	Introducing $\boldsymbol{x} = [ ..., q_{ij}, ..., s]^{T}$, $\boldsymbol{c} = [ 0, ..., 0, 1]^{T}$, $\boldsymbol{F}_{0}$, $\boldsymbol{F}_{i}$ as below,
\begin{equation}
	\label{eq:Eq20171114782}
	\begin{gathered}
			\boldsymbol{F}_{0}  =
			\left[ \begin{array}{ccc}
				{ \widetilde{\boldsymbol{J}} - \boldsymbol{I} }
				&{\boldsymbol{0}}
				&{\boldsymbol{0}}
				\\
				{\boldsymbol{0}}
				&{ -\widetilde{\boldsymbol{J}} + \boldsymbol{I} }
				&{\boldsymbol{0}}
				\\
				{\boldsymbol{0}}
				&{\boldsymbol{0}}
				&{\boldsymbol{D}}
			\end{array} \right],
			\;\;
			\boldsymbol{F}_{s}  =
			\left[ \begin{array}{ccc}
				{ \boldsymbol{I} }
				&{\boldsymbol{0}}
				&{\boldsymbol{0}}
				\\
				{\boldsymbol{0}}
				&{ \boldsymbol{I} }
				&{\boldsymbol{0}}
				\\
				{\boldsymbol{0}}
				&{\boldsymbol{0}}
				&{\boldsymbol{0}}
			\end{array} \right],
	\end{gathered}
\end{equation}
\begin{equation}
	\label{eq:Eq20171114834}
	\begin{aligned}
		\boldsymbol{F}_{i}  =
			\left[ \begin{array}{ccc}
				{ \boldsymbol{D}^{-\frac{1}{2}} \left( \boldsymbol{e}_{k} - \boldsymbol{e}_{l} \right) \left( \boldsymbol{e}_{k} - \boldsymbol{e}_{l} \right)^{T} \boldsymbol{D}^{-\frac{1}{2}} }
				&{\boldsymbol{0}}
				&{\boldsymbol{0}}
				\\
				{\boldsymbol{0}}
				&{ -\boldsymbol{D}^{-\frac{1}{2}} \left( \boldsymbol{e}_{k} - \boldsymbol{e}_{l} \right) \left( \boldsymbol{e}_{k} - \boldsymbol{e}_{l} \right)^{T} \boldsymbol{D}^{-\frac{1}{2}} }
				&{\boldsymbol{0}}
				\\
				{\boldsymbol{0}}
				&{\boldsymbol{0}}
				&{ - \boldsymbol{e}_{k}\boldsymbol{e}_{k}^{T} - \boldsymbol{e}_{l}\boldsymbol{e}_{l}^{T} }
			\end{array} \right].
	\end{aligned}
\end{equation}
for $i = 1, ..., |\boldsymbol{x}|-1$,
problem (\ref{eq:Eq201801262501}) can be written in the standard form of the semidefinite programming \cite{BoydConvexBook,JafarizadehIEEESensors2011} as below,
\begin{equation}
	\label{eq:Eq201801272653}
	\begin{aligned}
		\min\limits_{\boldsymbol{x}}
		\;\;
		&\boldsymbol{c}^{T} \cdot \boldsymbol{x},
		\\
		s.t.
		\quad
		&
		\boldsymbol{F}(x) = \sum\nolimits_{i=1}^{|\boldsymbol{x}|} \boldsymbol{x}_{i} \boldsymbol{F}_{i} + \boldsymbol{F}_{0} \succeq 0
	\end{aligned}
\end{equation}
The dual problem is as following,
\begin{equation}
	\label{eq:Eq201801272673}
	\begin{aligned}
		\max\limits_{\boldsymbol{Z}}
		\;\;
		&-Tr \left[ \boldsymbol{F}_{0} \cdot \boldsymbol{Z} \cdot \boldsymbol{Z}^{T} \right],
		\\
		s.t.
		\quad
		&
		Tr \left[ \boldsymbol{F}_{s} \cdot \boldsymbol{Z} \cdot \boldsymbol{Z}^{T} \right]  = \boldsymbol{c}_{|\boldsymbol{x}|} = 1,
		\\
		&
			Tr \left[ \boldsymbol{F}_{i} \cdot \boldsymbol{Z} \cdot \boldsymbol{Z}^{T} \right]  = \boldsymbol{c}_{i} = 0 \;\; \text{for} \;\; i=1, ..., |\boldsymbol{x}|-1.
	\end{aligned}
\end{equation}
The dual variable $\boldsymbol{Z}$ can be written as
$\boldsymbol{Z}  =
\left[ \boldsymbol{Z}_{1}^{T}, \boldsymbol{Z}_{2}^{T}, \boldsymbol{Z}_{3}^{T} \right]^{T}$
where
\begin{equation}
	\label{eq:Eq201801262582}
	\begin{gathered}
		\boldsymbol{Z}_{1}
		=
		\sum\nolimits_{ \{i,j\} \in \mathcal{E} } a_{ij} \boldsymbol{D}^{-\frac{1}{2}} \left( \boldsymbol{e}_{i} - \boldsymbol{e}_{j} \right),
		\\
			\boldsymbol{Z}_{2}
			=
			\sum\nolimits_{ \{i,j\} \in \mathcal{E} } b_{ij} \boldsymbol{D}^{-\frac{1}{2}} \left( \boldsymbol{e}_{i} - \boldsymbol{e}_{j} \right),
			\;\;
			\boldsymbol{Z}_{3}
			=
			\sum\nolimits_{ i \in \mathcal{V} } c_{i} \boldsymbol{e}_{i}
	\end{gathered}
\end{equation}
Using
(\ref{eq:Eq201801262582}), the dual constraints in (\ref{eq:Eq201801272673}) (i.e. $Tr\left( \boldsymbol{Z} \boldsymbol{F}_{i} \boldsymbol{Z}^{T} \right)  =  0$) reduce to the following
\begin{equation}
	\label{eq:Eq20171114916}
	\begin{gathered}
			\left( \boldsymbol{Z}_{1}^{T} \boldsymbol{D}^{-\frac{1}{2}} \left( \boldsymbol{e}_{i} - \boldsymbol{e}_{j} \right) \right)^{2}
			-
			\left( \boldsymbol{Z}_{2}^{T} \boldsymbol{D}^{-\frac{1}{2}} \left( \boldsymbol{e}_{i} - \boldsymbol{e}_{j} \right) \right)^{2}
			-
			\left( c_{i}^{2} + c_{j}^{2} \right)
			=
			0
	\end{gathered}
\end{equation}
The complementary slackness condition \cite{BoydConvexBook,JafarizadehIEEESensors2011}, reduces to
\begin{subequations}
	\label{eq:Eq20171114938}
	\begin{gather}
		\left( \left( s - 1 \right) \boldsymbol{I} + \boldsymbol{D}^{-\frac{1}{2}} \boldsymbol{L} \left( q \right) \boldsymbol{D}^{-\frac{1}{2}}  \right)  \boldsymbol{Z}_{1}  =  0,
		\label{eq:Eq20171114938a}
		\\
		\left( \left( s + 1 \right) \boldsymbol{I} - \boldsymbol{D}^{-\frac{1}{2}} \boldsymbol{L} \left( q \right) \boldsymbol{D}^{-\frac{1}{2}}  \right)  \boldsymbol{Z}_{2}  =  0,
		\label{eq:Eq20171114938b}
		\\
		c_{i} \left( \pi_{i} - \sum\nolimits_{j \neq i} q_{ij} \right)  =  0.
		\label{eq:Eq20171114938c}
	\end{gather}
\end{subequations}
\begin{remark}
	From (\ref{eq:Eq20171114916}), it is obvious that $\boldsymbol{Z_{1}} = \boldsymbol{0}$ results $\boldsymbol{Z_{2}} = \boldsymbol{0}$ and $c_{i} = 0$, which is not acceptable.
	Thus, it can be concluded that $\boldsymbol{Z_{1}}$ is always non-zero.
\end{remark}
\begin{theorem}
	\label{theorem-2-3450}
	If $\boldsymbol{Z_{1}}$ and $\boldsymbol{Z_{2}}$ are non-zero,
	then
	for optimal value of $s$,
	$1 - s$ and $s + 1$ are both eigenvalues of
	$\boldsymbol{D}^{-\frac{1}{2}}  \boldsymbol{L}(q)  \boldsymbol{D}^{-\frac{1}{2}}$,
	with corresponding eigenvectors $\boldsymbol{Z_{1}}$ and $\boldsymbol{Z_{2}}$.
\end{theorem}
\begin{proof}
	Since both $\boldsymbol{Z_{1}}$ and $\boldsymbol{Z_{2}}$ are non-zero, 
	then from
	(\ref{eq:Eq20171114938a}) 
	and
	(\ref{eq:Eq20171114938b})
	(which hold for optimal value of $s$, due to complementary slackness condition),
	it is obvious that 
	$1 - s$ and $s + 1$ are both eigenvalues of
	$\boldsymbol{D}^{-\frac{1}{2}}  \boldsymbol{L}(q)  \boldsymbol{D}^{-\frac{1}{2}}$,
	with corresponding eigenvectors $\boldsymbol{Z_{1}}$ and $\boldsymbol{Z_{2}}$.
\end{proof}
\begin{remark}
	\label{remark-sum-of-eigenvalues}
	Based on Theorem
	\ref{theorem-2-3450},
	for the optimal $SLEM$, it can be concluded that
	$s$ $=$ $\mu$ $(  \boldsymbol{D}^{-1/2}  \boldsymbol{L}(q)  \boldsymbol{D}^{-1/2}  )$ $=$ $1$ $-$ $\lambda_{2}$ $(  \boldsymbol{D}^{-1/2}  \boldsymbol{L}(q)  \boldsymbol{D}^{-1/2}  )$ $=$ $\lambda_{N}$ $(  \boldsymbol{D}^{-1/2}  \boldsymbol{L}(q)  \boldsymbol{D}^{-1/2}  )$ $-$ $1$.
	Thus we have 
	$\lambda_{2}(  \boldsymbol{D}^{-1/2}  \boldsymbol{L}(q)  \boldsymbol{D}^{-1/2}  )    +    \lambda_{N}(  \boldsymbol{D}^{-1/2}  \boldsymbol{L}(q)  \boldsymbol{D}^{-1/2}  ) = 2$.
\end{remark}
If the last constraints of (\ref{eq:Eq201801262501})
are strict (i.e. $- \sum\nolimits_{k \neq i} q_{ik} + \pi_{i} > 0$ for $i=1, ..., N$) then it can be concluded that $\pi_{i} - \sum_{j \neq i} q_{ij} \neq 0$.
Thus, we have
\begin{equation}
	\label{eq:Eq20171114955}
	\begin{gathered}
		c_{i}  =  0
		\;\;  \text{for}  \;\;  i=1, ..., N.
	\end{gathered}
\end{equation}
This will dictate additional constraints on the equilibrium distribution.
In general (\ref{eq:Eq20171114938}) hold for any underlying topology.
But in the case of tree topologies where vectors $\left( \boldsymbol{e}_{i} - \boldsymbol{e}_{j} \right)$ for $i \neq j$ are independent of each other,
based on Theorem
\ref{theorem-2-3450},
(\ref{eq:Eq20171114938a}) and (\ref{eq:Eq20171114938b}) can be interpreted as below,
\begin{subequations}
	\label{eq:Eq20171114964}
	\begin{gather}
		\left( s - 1 \right) a_{ij}  +  q_{ij} \left( \boldsymbol{e}_{i} - \boldsymbol{e}_{j} \right)^{T} \boldsymbol{D}^{-\frac{1}{2}} \boldsymbol{Z}_{1}  =  0,
		\label{eq:Eq20171114964a}
		\\
		\left( s + 1 \right) b_{ij}  -  q_{ij} \left( \boldsymbol{e}_{i} - \boldsymbol{e}_{j} \right)^{T} \boldsymbol{D}^{-\frac{1}{2}} \boldsymbol{Z}_{2}  =  0,
		\label{eq:Eq20171114964b}
	\end{gather}
\end{subequations}
Using the relations
(\ref{eq:Eq20171114916}) and
(\ref{eq:Eq20171114955}),
from (\ref{eq:Eq20171114964})
the following can be concluded
\begin{equation}
	\label{eq:Eq20171114977}
	\begin{gathered}
		\left( (s-1) a_{ij} \right)^{2}  =  \left( (s+1) b_{ij} \right)^{2}
	\end{gathered}
\end{equation}
for $\{i,j\} \in \mathcal{E}$.
In general topologies (with loops), (\ref{eq:Eq20171114964}) and (\ref{eq:Eq20171114977}) hold for edges which are not part of a loop.

\section{Optimal Transition Probabilities on Subgraphs of An Arbitrary Graph}
\label{sec:Branches}

In this section, we address the optimization of the FMMC problem (\ref{eq:Eq201712231257}) over different types of subgraphs where we determine the optimal transition probabilities over the edges of these subgraphs independent of the rest of topology.
The subgraphs considered in this section are connected to one node in the rest of topology.
Thus, $\boldsymbol{L}\left( q \right)$ for a subgraph with $N$ nodes can be written as below,
\begin{equation}
	\label{eq:Eq201808083130}
	\begin{gathered}
		\boldsymbol{L}(q)  =
		\left[ \begin{array}{ccc}
			{ \boldsymbol{L}_{B}\left( q \right) }
			&{ \boldsymbol{A}\left( q \right) }
			&{ \boldsymbol{0} }
			\\
			{ \boldsymbol{A}^{T}\left( q \right) }
			&{ \boldsymbol{L}_{N+1}\left( q \right) }
			&{ \boldsymbol{B}^{T}\left( q \right) }
			\\
			{ \boldsymbol{0} }
			&{ \boldsymbol{B}\left( q \right) }
			&{ \boldsymbol{L}^{'}\left( q \right) }
		\end{array} \right],
	\end{gathered}
\end{equation}
where $\boldsymbol{L}_{B}\left( q \right)$ is a $N \times N$ submatrix defined based on the topology of the subgraph and $\boldsymbol{A}\left( q \right)$ is a $N \times 1$ column vector in the form of $\boldsymbol{A}\left( q \right)^{T}  =  \left[ 0, ..., 0, q_{N,N+1} \right]$
which denotes
the connection between the subgraph and the $(N+1)$-th node in the whole topology.
$\boldsymbol{L}_{N+1}\left( q \right)$ is a $1 \times 1$ block of the matrix corresponding to the $(N+1)$-th node.
The square submatrix $\boldsymbol{L}^{'}\left( q \right)$ and the column vector $\boldsymbol{B} \left( q \right)$ correspond to the rest of graph indicating the unknown part of the topology.
Accordingly the operator $\boldsymbol{D}^{-\frac{1}{2}} \boldsymbol{L}\left( q \right) \boldsymbol{D}^{-\frac{1}{2}}$ for an arbitrary subgraph can be written as below
\begin{equation}
	\nonumber
	\begin{gathered}
			\left[ \begin{array}{ccc}
				{  \boldsymbol{D}_{B}^{-\frac{1}{2}} \boldsymbol{L}_{B} \left( q \right) \boldsymbol{D}_{B}^{-\frac{1}{2}}  }
				&{  \boldsymbol{D}_{B}^{-\frac{1}{2}} \boldsymbol{A}\left( q \right)  \pi_{N+1}^{-\frac{1}{2}}  }
				&{ \boldsymbol{0} }
				\\
				{ \pi_{N+1}^{-\frac{1}{2}} \boldsymbol{A}^{T}\left( q \right)  \boldsymbol{D}_{B}^{-\frac{1}{2}} }
				&{ \pi_{N+1}^{-\frac{1}{2}} \boldsymbol{L}_{N+1}\left( q \right)  \pi_{N+1}^{-\frac{1}{2}} }
				&{ \pi_{N+1}^{-\frac{1}{2}} \boldsymbol{B}^{T}\left( q \right)  {\boldsymbol{D}^{'}}^{-\frac{1}{2}} }
				\\
				{ \boldsymbol{0} }
				&{ {\boldsymbol{D}^{'}}^{-\frac{1}{2}} \boldsymbol{B}\left( q \right)  \pi_{N+1}^{-\frac{1}{2}} }
				&{ {\boldsymbol{D}^{'}}^{-\frac{1}{2}} \boldsymbol{L}^{'}\left( q \right)  {\boldsymbol{D}^{'}}^{-\frac{1}{2}} }
			\end{array} \right],
	\end{gathered}
\end{equation}
$\boldsymbol{D}_{B} = diag\left( \left[ \pi_{1}, \pi_{2}, ..., \pi_{N} \right] \right)$ is a $N \times N$ diagonal submatrix corresponding to the arbitrary subgraph part of the topology and
$\boldsymbol{D}^{'}$ corresponds to the rest of graph indicating the unknown part of the topology.
We define vectors $\boldsymbol{e}_{i}$ for $i \in \mathcal{V}$ as column vectors with $i$-th element equal to $1$ and zero elsewhere.
The size of vectors $\boldsymbol{e}_{i}$ for $i \in \mathcal{V}$ is same as the number of vertices in the whole graph.
These vectors can be divided into three parts i.e.
$\boldsymbol{e}_{i}^{T} = \left[ {\boldsymbol{e}_{B}}_{i}^{T}, {\boldsymbol{e}_{i}} ( N+1 ), {\boldsymbol{e}^{'}}_{i}^{T} \right]$,
where ${\boldsymbol{e}_{B}}_{i}^{T}$ is the first $N$ vertices (corresponding to vertices in the subgraph) and ${\boldsymbol{e}^{'}}_{i}^{T}$ corresponds to the other vertices which are in the unknown part of the topology.
For $i=1, ..., N$, the $i$-th element of ${\boldsymbol{e}_{B}}_{i}$ is one and the rest is zero,
${\boldsymbol{e}_{i}} ( N+1 )$
is zero and ${\boldsymbol{e}^{'}}_{i}$ is a vector of all zeros.
For $i=N+1$, ${\boldsymbol{e}_{B}}_{i}$ and ${\boldsymbol{e}^{'}}_{i}$ is a vector of all zeros and
${\boldsymbol{e}_{i}} ( N+1 )$
is one.
For $i=N+1, ..., \left| \mathcal{V} \right|$, the $(i-N)$-th element of ${\boldsymbol{e}^{'}}_{i}$ is one and the rest is zero,
${\boldsymbol{e}_{i}} ( N+1 )$
is zero and ${\boldsymbol{e}_{B}}_{i}$ is a vector of all zeros.
Similarly, we divide the dual variable into three parts as below,
$\boldsymbol{Z}_{1}^{T}$    $=$    $[$    ${\boldsymbol{Z}_{1}}_{B}^{T}$,    $- a_{N,N+1} \pi_{N+1}^{-1/2}$    $+$    $\sum\nolimits_{ j \in \mathcal{N}_{N+1}^{'} } a_{N+1,j} \pi_{N+1}^{-1/2}$,    $\boldsymbol{Z}_{1}^{'T}$    $]$,
$\boldsymbol{Z}_{2}^{T}$        $=$        $[$        ${\boldsymbol{Z}_{2}}_{B}^{T}$,        $- b_{N,N+1} \pi_{N+1}^{-\frac{1}{2}}$        $+$        $\sum\nolimits_{ j \in \mathcal{N}_{N+1}^{'} }        b_{N+1,j} \pi_{N+1}^{ -\frac{1}{2} }$,        $\boldsymbol{Z}_{2}^{'T}$        $]$,
$\boldsymbol{Z}_{3}$    $=$    ${\boldsymbol{Z}_{3}}_{B}$    $+$    $c_{N+1} \boldsymbol{e}_{N+1}$    $+$    $\sum\nolimits_{ i \in \mathcal{V}^{'} } c_{i} \boldsymbol{e}_{i}$,
where
$\mathcal{N}_{N+1}^{'}$ is the set of neighbors of $(N+1)$-th node within $\mathcal{E}^{'}$
and
$\boldsymbol{Z}_{1}^{'}$  $=$  $\sum_{ \{i,j\} \in \mathcal{E}^{'} }$ $a_{ij}$ ${\boldsymbol{D}^{'}}^{-\frac{1}{2}}$ $($ ${\boldsymbol{e}^{'}}_{i}$ $-$ ${\boldsymbol{e}^{'}}_{j}$ $)$  $-$ $\sum_{ j \in \mathcal{N}_{N+1}^{'} }$  $a_{N+1,j}$ $\pi_{j}^{-\frac{1}{2}}$ ${\boldsymbol{e}^{'}}_{j}$,
and
$\boldsymbol{Z}_{2}^{'}$  $=$  $\sum_{ \{i,j\} \in \mathcal{E}^{'} }$ $b_{ij}$ ${\boldsymbol{D}^{'}}^{-\frac{1}{2}}$ $($ ${\boldsymbol{e}^{'}}_{i}$ $-$ ${\boldsymbol{e}^{'}}_{j} )$  $-$ $\sum_{ j \in \mathcal{N}_{N+1}^{'} }$  $b_{N+1,j}$ $\pi_{j}^{-\frac{1}{2}}$ ${\boldsymbol{e}^{'}}_{j}$.
\begin{equation}
	\label{eq:Eq201712021281}
	\begin{gathered}
			{\boldsymbol{Z}_{1}}_{B}
			=
			\sum\nolimits_{ \{i,j\} \in \mathcal{E}_{B} } a_{ij} \boldsymbol{D}_{B}^{-\frac{1}{2}} \left( {\boldsymbol{e}_{B}}_{i} - {\boldsymbol{e}_{B}}_{j} \right)
			+ a_{N,N+1} \pi_{N}^{-\frac{1}{2}} {\boldsymbol{e}_{B}}_{N}
	\end{gathered}
\end{equation}
\begin{equation}
	\label{eq:Eq201712021292}
	\begin{gathered}
			{\boldsymbol{Z}_{2}}_{B}
			=
			\sum\nolimits_{ \{i,j\} \in \mathcal{E}_{B} } b_{ij} \boldsymbol{D}_{B}^{-\frac{1}{2}} \left( {\boldsymbol{e}_{B}}_{i} - {\boldsymbol{e}_{B}}_{j} \right)
			+ b_{N,N+1} \pi_{N}^{-\frac{1}{2}} {\boldsymbol{e}_{B}}_{N}
	\end{gathered}
\end{equation}
and
${\boldsymbol{Z}_{3}}_{B}    =    \sum\nolimits_{ i \in \mathcal{V}_{B} } c_{i} \boldsymbol{e}_{i}$
with
$\mathcal{V}_{B} = \{ 1, ..., N \}$
and
$\mathcal{V}^{'}$ as the set of nodes in the unknown part of the topology.
Based on the above derivation of the dual variable $\boldsymbol{Z}$, and vectors $\boldsymbol{e}_{i}$,
equations (\ref{eq:Eq20171114916}) and (\ref{eq:Eq20171114938}) for the arbitrary subgraph part of the topology can be writen as below,
\begin{equation}
	\label{eq:Eq201712011530}
	\begin{gathered}
			\left( {\boldsymbol{Z}_{1}}_{B}^{T} \boldsymbol{D}_{B}^{-\frac{1}{2}}
			{\overline{\boldsymbol{e}}_{B}}_{ij}
			\right)^{2}
			-
			\left( {\boldsymbol{Z}_{2}}_{B}^{T} \boldsymbol{D}_{B}^{-\frac{1}{2}}
			{\overline{\boldsymbol{e}}_{B}}_{ij}
			\right)^{2}
			-
			\left( c_{i}^{2} + c_{j}^{2} \right)
			=
			0
	\end{gathered}
\end{equation}
with
${\overline{\boldsymbol{e}}_{B}}_{ij}  =  {\boldsymbol{e}_{B}}_{i} - {\boldsymbol{e}_{B}}_{j}$
for
$i,j \in \mathcal{V}_{B}$ and
\begin{subequations}
	\label{eq:Eq201711251504}
	\begin{gather}
			\left( \left( s - 1 \right) \boldsymbol{I}_{B} + \boldsymbol{D}_{B}^{-\frac{1}{2}} \boldsymbol{L}_{B}(q) \boldsymbol{D}_{B}^{-\frac{1}{2}}  \right)
			{\boldsymbol{Z}_{1}}_{B}
			+
			\frac{ q_{N,N+1} }  { \pi_{N+1} \sqrt{\pi_{N}} }
			\left(
			-
			a_{N,N+1}
			+
			\sum\nolimits_{ j \in \mathcal{N}_{N+1}^{'} } a_{N+1,j}
			\right)
			=
			0,
		\label{eq:Eq201711251504a}
		\\
			\left( \left( s + 1 \right) \boldsymbol{I}_{B} - \boldsymbol{D}_{B}^{-\frac{1}{2}} \boldsymbol{L}_{B}\left( q \right) \boldsymbol{D}_{B}^{-\frac{1}{2}}  \right)
			{\boldsymbol{Z}_{2}}_{B}
			+
			\frac{ q_{N,N+1} }  { \pi_{N+1} \sqrt{\pi_{N}} }
			\left(
			b_{N,N+1}
			-
			\sum\nolimits_{ j \in \mathcal{N}_{N+1}^{'} } b_{N+1,j}
			\right)
			=
			0,
		\label{eq:Eq201711251504b}
		\\
			c_{i} \left( \pi_{i} - \sum_{j \neq i} q_{ij} \right)  =  0,
			\;\; \text{for} \;\; i \in \{1,...,N\}, j \in \{1,...,N+1\}
		\label{eq:Eq201711251504c}
	\end{gather}
\end{subequations}
In general, the optimal value of $q_{i,j}$ for edges of the subgraph are obtained from solving equations (\ref{eq:Eq201712011530}) and (\ref{eq:Eq201711251504}) (excluding the equations corresponding to $N$-th node).

\subsection{Palm Subgraph}
\label{sec:palmbranch}
This subgraph is formed by $n$ path subgraphs of length one connected to the free end of a Path subgraph of length $m$, where $m \geq 2$.
For this subgraph $\boldsymbol{L}_{B}\left( q \right)$ can be written as below,
\begin{equation}
	\label{eq:Eq201711181364-SB}
	\begin{aligned}
		\boldsymbol{L}_{B}\left( q \right)
		&=
		\sum\nolimits_{i=-1}^{-n} q_{i} \left( {\boldsymbol{e}_{B}}_{i} - {\boldsymbol{e}_{B}}_{0} \right)  \left( {\boldsymbol{e}_{B}}_{i} - {\boldsymbol{e}_{B}}_{0} \right)^{T}
			+
			\sum\nolimits_{i=0}^{m-1} q_{i} \left( {\boldsymbol{e}_{B}}_{i} - {\boldsymbol{e}_{B}}_{i+1} \right)  \left( {\boldsymbol{e}_{B}}_{i} - {\boldsymbol{e}_{B}}_{i+1} \right)^{T}
	\end{aligned}
\end{equation}
where ${\boldsymbol{e}_{B}}_{0}$ corresponds to the vertex at the free end of the path subgraph that the star part is formed.
Following a similar solution as the one
presented in Appendix \ref{sec:StarAbstractSolution} for Star topology it can be shown
that the optimal value of $q_{i}$ is as below,
\begin{equation}
	\label{eq:201808084378}
	\begin{gathered}
			q_{i} = ( \pi_{0} \pi_{i} ) / ( \pi_{0} + \sum\nolimits_{j=-1}^{-n} \pi_{j} ),
			\;\; \text{for} \;\; i=-1,...,-n,
	\end{gathered}
\end{equation}
\begin{equation}
	\label{eq:201808084389}
	\begin{gathered}
		q_{0} = ( \pi_{0} \pi_{1} ) / ( \pi_{0} + \pi_{1} ),
	\end{gathered}
\end{equation}
if
$\pi_{0}^{2}
\geq
\pi_{1}
\sum_{j=-1}^{-n} \pi_{j}$
and
$\pi_{1} \geq \sum_{j=-1}^{-n} \pi_{j} $,
where the first inequality originates from the constraint
$\pi_{0} > q_{0} + \sum_{j=-1}^{-n} q_{j}$,
and the second inequality originates from
the observation that
and arbitrary graph with a palm subgraph (with $m \geq 2$) should have a larger $SLEM$ than that of star topology with path subgraphs of length one, which in turn is equivalent to a Palm subgraph with $m=2$ which is connected to a single vertex.
The optimal value of $q$ over edges of the path part are
\begin{equation}
	\label{eq:201808084419}
	\begin{gathered}
		q_{i} = ( \pi_{i} \pi_{i+1} ) / ( \pi_{i} + \pi_{i+1} ),
	\end{gathered}
\end{equation}
if
$\pi_{i}^{2} \geq \pi_{i-1} \pi_{i+1}$,
for $i=1,...,m-2$.

\begin{remark}
	\label{Remark-StarBranch}
	An interesting point is that in the star part of the palm subgraph if
	(\ref{eq:Eq201801262501c}) is satisfied with strict inequality,
	the followings can be concluded.
	$i)$
	Considering $q_{i} = \pi_{i} \boldsymbol{P}_{i,0}$, from $\frac{q_{i}}{\pi_{i}} = \frac{q_{j}}{\pi_{j}}$, (which is derived from (\ref{eq:201808084378})), it can be concluded that $\boldsymbol{P}_{i,0} = \boldsymbol{P}_{j,0}$. i.e., the transition probabilities from every vertex in subgraphs to central vertex are all equal.
	$ii)$
	Considering $q_{i} = \pi_{0} \boldsymbol{P}_{0,i}$, from $\frac{q_{i}}{\pi_{i}} = \frac{q_{j}}{\pi_{j}}$, it can be concluded that $\boldsymbol{P}_{0,i} = \frac{\mu}{\pi_{0}} \pi_{i}$. i.e., the transition probabilities from the central vertex to $i$-th vertex in subgraphs is linearly proportional to value of equilibrium distribution in $i$-th vertex.
\end{remark}

\subsection{Path Subgraph}
\label{sec:pathbranch}
This subgraph is formed by a path graph (with $N$ nodes).
For this subgraph $\boldsymbol{L}_{B}\left( q \right)$ can be written as below,
\begin{equation}
	\label{eq:Eq201711181364}
	\begin{gathered}
		\boldsymbol{L}_{B}\left( q \right)
		=
		\sum\nolimits_{i = 1}^{N-1} q_{i} \left( {\boldsymbol{e}_{B}}_{i} - {\boldsymbol{e}_{B}}_{i+1} \right)  \left( {\boldsymbol{e}_{B}}_{i} - {\boldsymbol{e}_{B}}_{i+1} \right)^{T}
	\end{gathered}
\end{equation}
	In Appendix \ref{sec:pathbranch-Appendix}, it is shown that the optimal value of $q_i$ is as below,
\begin{equation}
	\label{eq:Eq201711171267}
	\begin{gathered}
		q_i
		=
		( \pi_i\pi_{i+1} ) / ( \pi_i+\pi_{i+1} ),
		\;\; \text{for} \;\; i=1,...,N-1
	\end{gathered}
\end{equation}
given that
\begin{equation}
	\label{eq:Eq201711171284}
	\begin{gathered}
		\pi_i^2\geq \pi_{i-1}\pi_{i+1}
		\;\; \text{for} \;\; i = 2, ..., N-2.
	\end{gathered}
\end{equation}
Both the Poisson distribution (i.e. $\pi_n=a\frac{\bar{n}^n}{n!}$) and the Binary distribution (i.e. $\pi_n=a\frac{M!}{n!(M-n)!}p^n(1-p)^{M-n}$, for $M>2N$, satisfy (\ref{eq:Eq201711171284}).

The optimal results presented in the rest of this section are for subgraphs which are clique lifted of Path subgraph.

\subsection{Semi-complete Subgraph}
\label{sec:Semi-complete-branch}

A semi-complete subgraph consists of a complete graph connected to two path graphs, such that the nodes at the end of path graphs are connected to all of the nodes in the complete graph part.
Semi-complete subgraph is identified by three parameters, namely $m, n_{1}$ and $n_{2}$.
$n_1$ and $n_2$ are the number of nodes in path graphs and $m$ is the number of nodes in the complete graph part.
A semi-complete subgraph with parameters $n_{1}=2, n_{2}=3$ and $m=3$ is depicted in Figure \ref{fig:SemiCompleteBranch}.
\begin{figure}
	\centering
	\includegraphics[width=160mm]{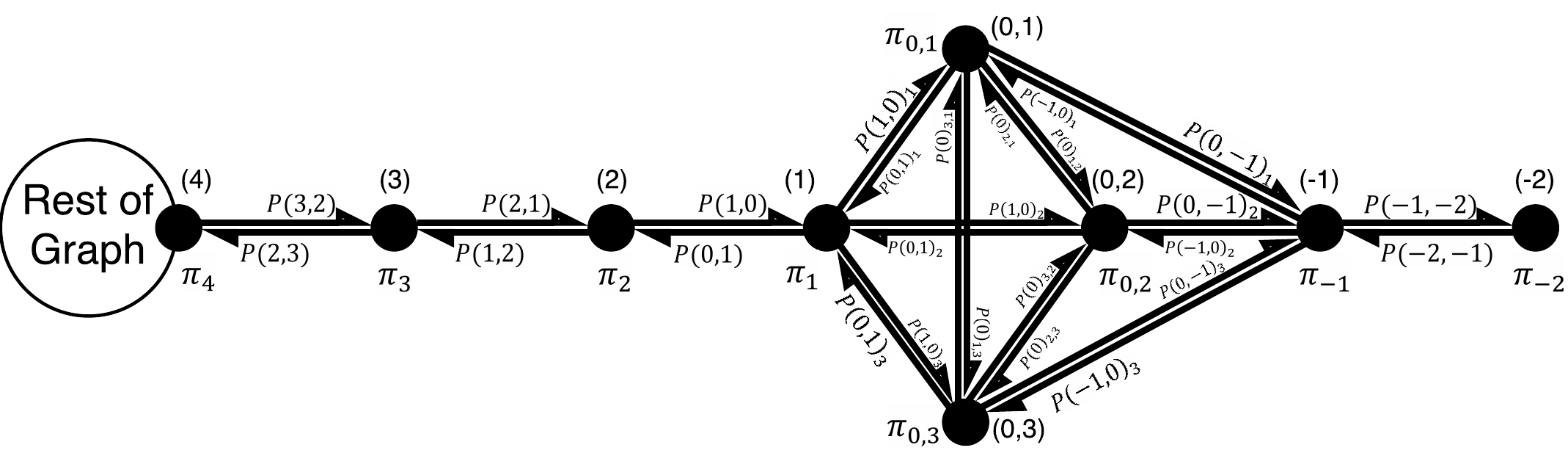}
	\caption{Semi-complete subgraph with $n_{1}=2, n_{2}=3$ and $m=3$.}
	\label{fig:SemiCompleteBranch}
\end{figure}
In semi-complete subgraph, the nodes in the first path graph (with length $n_1$) are denoted by $\{(-n_1), ..., (-1)\}$ where node $(-1)$ is connected to all nodes in the complete graph part.
The nodes in the complete graph part are denoted by $\{(0,1), ..., (0,m)\}$ and nodes in the second path graph (with length $n_2$) are denoted by $\{(1), ..., (n_2)\}$ where node $(1)$ is connected to all nodes in the complete graph part and node $(n_2)$ is connected to node $(N+1)$ in the unknown part of the topology.
Semi-complete subgraph is the clique lift of a path subgraph where the fiber of the vertex $(0)$ in the middle of the path subgraph is a clique with $m$ vertices and the rest of the vertices stay untouched.
Thus, for the weights on edges between nodes $(0,\alpha)$ and $(1)$, we use $q(0,1)_{\alpha}$ for $\alpha=1, ..., m$,
and similarly for the edges between nodes $(0,\alpha)$ and $(-1)$, we use $q(-1,0)_{\alpha}$ for $\alpha=1, ..., m$.
For the weights on edges between nodes $(i)$ and $(i+1)$, we use $q(i+1, i)$ for $i=-2, ..., -n_1$,
and for the weights on edges between nodes $(i-1)$ and $(i)$, we use $q(i-1, i)$ for $i=2, ..., n_2$.
For the edges between nodes $(0, \alpha)$ and $(0, \beta)$, we use $q(0)_{\alpha, \beta}$ for $\alpha \neq \beta =1, ..., m$.
By $\pi_{i}$, we denote the equilibrium distribution over vertex $(i)$ for $i \in \{-1, ..., -n_{1}\} \cup \{1, ..., n_{2}\}$, and by $\pi_{0,\alpha}$, we denote the equilibrium distribution over vertex $(0, \alpha)$ for $\alpha = 1, ..., m$.
The optimal value of $q(i+1, i)$ is equal to
$\frac{ \pi_{i} \pi_{i+1} }{ \pi_{i} + \pi_{i+1} }$ for $i = -2, ..., -n_{1}$
and
$q(i-1, i)$ is equal to
$\frac{ \pi_{i-1} \pi_{i} }{ \pi_{i-1} + \pi_{i} }$ for $i = 2, ..., n_{2}$.
The optimal values of $q(-1,0)_{\alpha}$ and $q(0,1)_{\alpha}$ are equal to $q(-1,0)_{\alpha} =  \frac{\pi_{-1} \pi_{0, \alpha}}{ \sum_{\gamma=1}^{m} \pi_{0, \gamma} + \pi_{-1} }   $ and $q(0,1)_{\alpha} = \frac{  \pi_{1} \pi_{0, \alpha}  }{   \sum_{\gamma=1}^{m} \pi_{0, \gamma} + \pi_{1} } $, for $\alpha = 1, ..., m$, repectively.
The optimal value of $q(0)_{\alpha, \beta}$ is
$\frac{ \pi_{0, \alpha} \pi_{0, \beta} (( \sum_{\gamma=1}^{m} \pi_{0, \gamma} )^{2} - \pi_{-1} \pi_{1}) }{   ( \sum_{\gamma=1}^{m} \pi_{0, \gamma} ) (\sum_{\gamma=1}^{m} \pi_{0, \gamma} + \pi_{1} ) ( \sum_{\gamma=1}^{m} \pi_{0, \gamma} + \pi_{-1} )} $
for $\alpha \neq \beta = 1, ..., m$.
These results are valid if the given equilibrium distribution $\pi_{i}$ satisfies (\ref{eq:Eq201711171284}) for $i = -n_{1}, ..., -2$ and $i=2, ..., n_{2}$
along with $\pi_{-1}^{2} \geq \pi_{-2} ( \sum_{\gamma=1}^{m} \pi_{0, \gamma} )$ and $\pi_{1}^{2} \geq \pi_{2} ( \sum_{\gamma=1}^{m} \pi_{0, \gamma} )$.
From $q_{0} \geq 0$, the following constraint is obtained $( \sum_{\gamma=1}^{m} \pi_{0, \gamma} )^{2} \geq \pi_{-1} \pi_{1}$.

\subsection{Extended Complete Ladder Subgraph}
\label{sec:ExtendedLadderBranch}

The extended complete ladder subgraph is a complete ladder graph with parameters $(m_{1}, m_{2}, ..., m_{n} )$, such that all nodes in the $n$-th layer are connected to the rest of graph.
An extended complete ladder graph ${ECL}_{n}$ is $ n$ complete graphs  with $ m_{1}, m_{2}, ..., m_{n}$ nodes each  connected to each other in the form of a $ n $-partite graph.
An extended complete ladder subgraph with parameters $m_{1} = 2$, $m_{2} = 3$, $m_{3} = 1$, $m_{4} = 2$ is depicted in Fig. \ref{fig:ExtendedCompleteLadderBranch}.
\begin{figure}
	\centering
	\includegraphics[width=110mm]{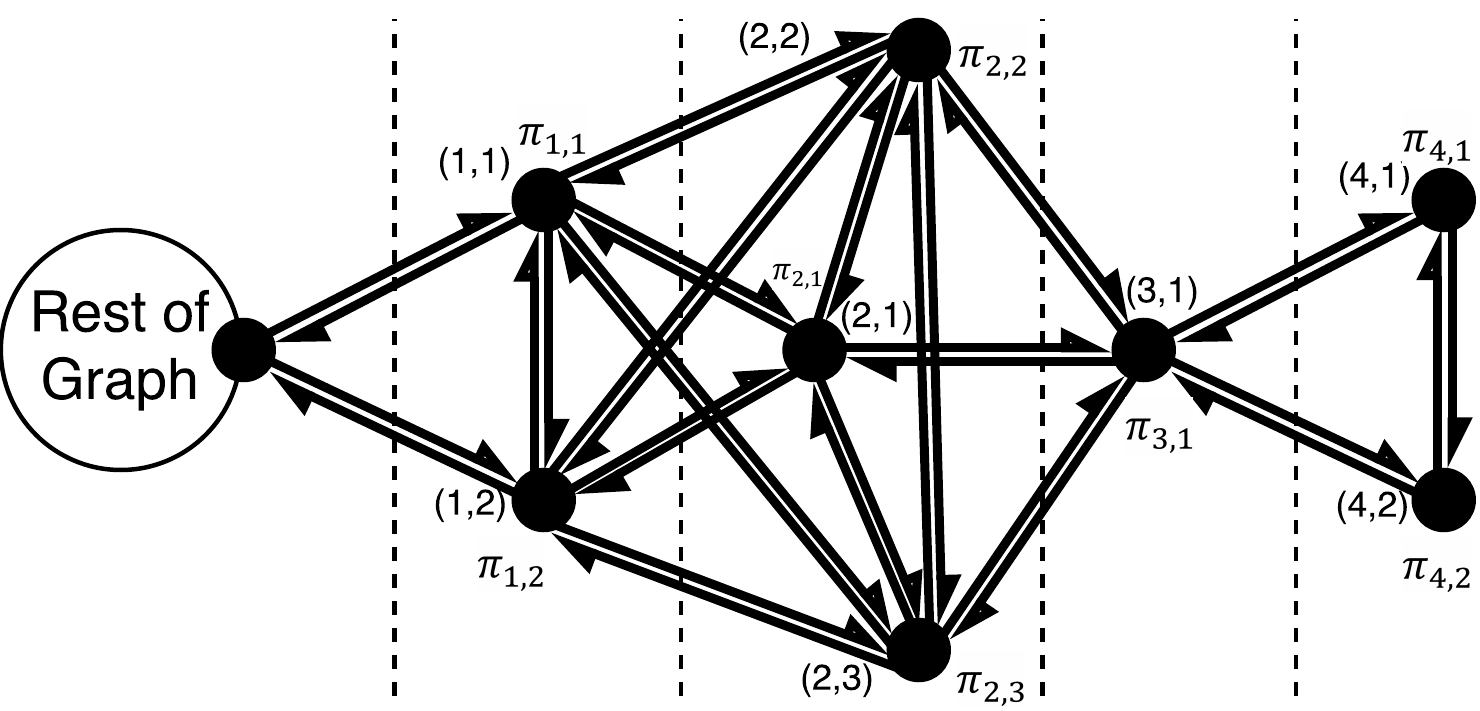}
	\caption{Extended complete ladder subgraph with parameters $m_{1} = 2$, $m_{2} = 3$, $m_{3} = 1$, $m_{4} = 2$.}
	\label{fig:ExtendedCompleteLadderBranch}
\end{figure}
In extended complete ladder subgraph, each vertex is denoted by $(i, \alpha_{i})$ where $i$ varies from $1$ to $n$ and
$\alpha_{i} = 1, ..., m_{i}$ for $i = 1, ..., n$.
Vertices $(n,\alpha)$ for $\alpha = 1, ..., m_{n}$ are connected to node $(N+1)$ in the unknown part of the topology, where $N = \sum_{i=1}^{n} m_{i}$ is the number of vertices in the subgraph.
Extended complete ladder subgraph is the clique lift of a path subgraph where the fiber of the vertices $(i)$ for $i=1, ..., n$ is a clique with $m_{i}$ vertices.
We use
$q(i, i+1)_{\alpha_{i}, \alpha_{i+1}}$ for denoting the weights on the edge $\{(i, \alpha_{i}), (i+1, \alpha_{i+1})\}$ for $i=1, ..., n-1$, $\alpha_{i} = 1, ..., m_{i}$ and $\alpha_{i+1} = 1, ..., m_{i+1}$,
and
$q(i)_{\alpha_{i}, \beta_{i}}$ for denoting the weights on the edge $\{(i, \alpha_{i}), (i, \beta_{i})\}$ for $i=1, ..., n$, $\alpha_{i} \neq \beta_{i} = 1, ..., m_{i}$.
By $\pi_{i, \alpha_{i}}$, we denote the equilibrium distribution over vertex $(i, \alpha_{i})$ for $i = 1, ...,  n$ and $\alpha_{i} = 1, ..., m_{i}$.
The optimal value of
$q(1)_{\alpha_{1}, \beta_{1}}$
is equal to
$\frac{\pi_{1,\alpha_1}\pi_{1,\beta_1} }{\sum_{\gamma=1}^{m_1} \pi_{1, \gamma}+\sum_{\gamma=1}^{m_2} \pi_{2, \gamma}}$
for $\alpha_{1} \neq \beta_{1} = 1, ..., m_{1}$.
The optimal value of
$q(i)_{\alpha_{i}, \beta_{i}}$
is equal to
$\frac{\pi_{i,\alpha_i}\pi_{i,\beta_i}((\sum_{\gamma=1}^{m_i} \pi_{i, \gamma})^2-(\sum_{\gamma=1}^{m_{i-1}} \pi_{i-1, \gamma})(\sum_{\gamma=1}^{m_{i+1}} \pi_{i+1, \gamma})) }{(\sum_{\gamma=1}^{m_i} \pi_{i, \gamma})(\sum_{\gamma=1}^{m_i} \pi_{i, \gamma}+\sum_{\gamma=1}^{m_{i-1}} \pi_{i-1, \gamma})(\sum_{\gamma=1}^{m_i} \pi_{i, \gamma}+\sum_{\gamma=1}^{m_{i+1}} \pi_{i+1, \gamma})}$
given that
$(\sum_{\gamma=1}^{m_i} \pi_{i, \gamma})^2  \geq  (\sum_{\gamma=1}^{m_{i-1}} \pi_{i-1, \gamma})(\sum_{\gamma=1}^{m_{i+1}} \pi_{i+1, \gamma})$
for $i=2, ..., n-1$ and $\alpha_{i} \neq \beta_{i} = 1, ..., m_{i}$.
The optimal value of
$q(i,i+1)_{\alpha_{i}, \alpha_{i+1}}$
is equal to
$\frac{\pi_{i,\alpha_{i}}\pi_{i+1,\alpha_{i+1}} }{\sum_{\gamma=1}^{m_i} \pi_{i, \gamma}+\sum_{\gamma=1}^{m_{i+1} } \pi_{i+1, \gamma}}$
for $i=1, ..., n-1$, $\alpha_{i} = 1, ..., m_{i}$ and $\alpha_{i+1} = 1, ..., m_{i+1}$.

\subsection{Lollipop Subgraph}
\label{sec:LollipopBranch}
The $(m,n)$-Lollipop graph $L_{(m,n)}$ is the graph obtained by joining a complete graph $K_m$ to a path graph $P_n$.
The Lollipop subgraph $L_{(m,n)}$ is defined as a Lollipop graph which is connected to an arbitrary graph via an edge connected to the end node of the path part as shown in Fig. \ref{fig:LollipopBranch} for $n=2$, $m=5$.
\begin{figure}
	\centering
	\includegraphics[width=120mm]{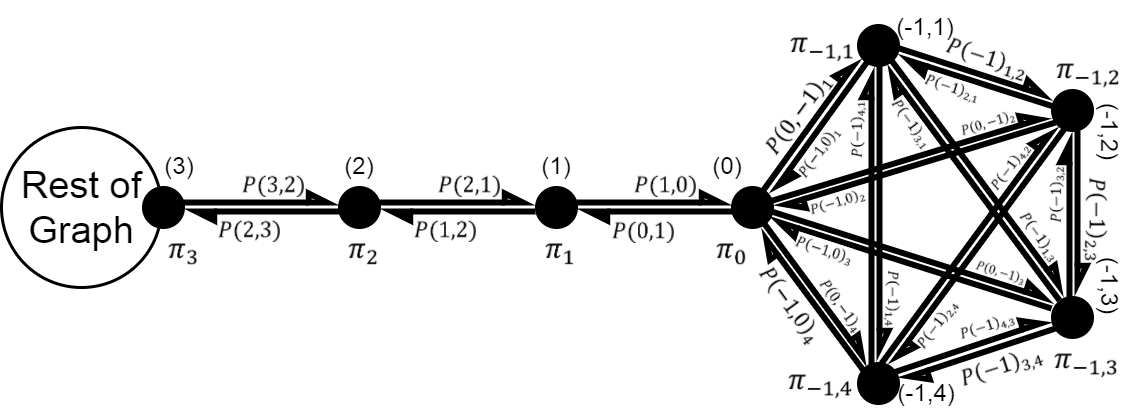}
	\caption{Lollipop subgraph with parameters $n=2$ and $m=5$.}
	\label{fig:LollipopBranch}
\end{figure}
The Lollipop subgraph $L_{(m,n)}$ is a special case of extended complete ladder subgraph with parameters $m_{1} = m_{2} = \cdots = m_{n-1} = 1$ and $m_{n} = m$.

\subsection{Example}

Here, we provide an example with all five subgraphs mentioned above.
In this example, the optimal value of $q$ are obtained numerically, where we compare them to the analytical closed-form formulas provided above.
The topology considered for this example is depicted in Fig.
\ref{fig:Example-Branches}.
The equilibrium distribution considered for this example is
$\pi_{i} = i$ for $i=1,...,7, 12, 13, 18, ..., 23, 27, 31$,
$\pi_{8} = 0.8$, $\pi_{9} = 0.9$, $\pi_{10} = 1$, $\pi_{11} = 1.1$,
$\pi_{14} = 1.4$, $\pi_{15} = 1.5$, $\pi_{16} = 1.6$, $\pi_{17} = 1.7$,
$\pi_{24} = 2.4$, $\pi_{25} = 2.5$,
$\pi_{26} = 2.6$,
$\pi_{28} = 11.2$, $\pi_{29} = 11.6$, $\pi_{30} = 12$,
$\pi_{32} = 3.2$.
The optimal value of $q$ for the Path subgraph are 
$q_{1,2} = 0.058191$, 
$q_{2,3} = 1.2 $,
$q_{3,4} = 1.714286$,
$q_{4,5} = 2.222$.
The optimal value of $q$ for the Palm subgraph are 
$q_{1,6} = 0.068084$, 
$q_{6,7} = 3.230788$, $q_{7,8} = 0.518532$, $q_{7,9} = 0.583355$, $q_{7,10} = 0.740753$, $q_{7,11} = 0.750025$.
The optimal value of $q$ for the Lollipop subgraph are 
$q_{1,12} = 0.126473$, 
$q_{12,13} = 6.240006$, $q_{13,14} = 0.947944$, $q_{13,15} = 1.015646$, $q_{13,16} = 1.083346$, $q_{13,17} = 1.151045$, $q_{14,15} = 0.109411$, $q_{14,16} = 0.116697$, $q_{14,17} = 0.123986$, $q_{15,16} = 0.125025$, $q_{15,17} = 0.132831$, $q_{16,17} = 0.141676$.
The optimal value of $q$ for the extended complete ladder subgraph are 
$q_{1,18} = 0.251854$, $q_{1,19} = 0.266908$, 
$q_{18,19} = 5.565850$, $q_{18,20} = 3.599990$, $q_{18,21} = 3.779993$, $q_{18,22} = 3.959996$, $q_{19,20} = 3.799989$, $q_{19,21} = 3.989992$, $q_{19,22} = 4.179995$, $q_{20,21} = 2.417034$, $q_{20,22} = 2.532134$, $q_{20,23} = 5.348812$, $q_{21,22} = 2.658744$, $q_{21,23} = 5.616256$, $q_{22,23} = 5.883701$, $q_{23,24} = 1.978508$, $q_{23,25} = 2.060948$, $q_{24,25} = 0.178161$.
The optimal value of $q$ for the Semi-complete subgraph are 
$q_{1,26} = 0.228464$, 
$q_{26,27} = 2.371536$, $q_{27,28} = 4.893188$, $q_{27,29} = 5.067973$, $q_{27,30} = 5.242719$, $q_{28,29} = 0.343425$, $q_{28,30} = 0.355233$, $q_{28,31} = 5.276578$, $q_{29,30} = 0.367962$, $q_{29,31} = 5.404536$, $q_{30,31} = 5.653494$, $q_{31,32} = 2.900599$.
Note that the numerical results obtained for all edges (except those connected to the central vertex $1$) are in agreement with the analytical results provided in
subsections
\ref{sec:pathbranch},
\ref{sec:palmbranch},
\ref{sec:LollipopBranch},
\ref{sec:ExtendedLadderBranch}
and
\ref{sec:Semi-complete-branch}.
The optimal $SLEM$ obtained for the topology depicted in Fig.
\ref{fig:Example-Branches} 
is equal to $0.99748869$.

\begin{figure}
	\centering
	\includegraphics[width=0.6\hsize]{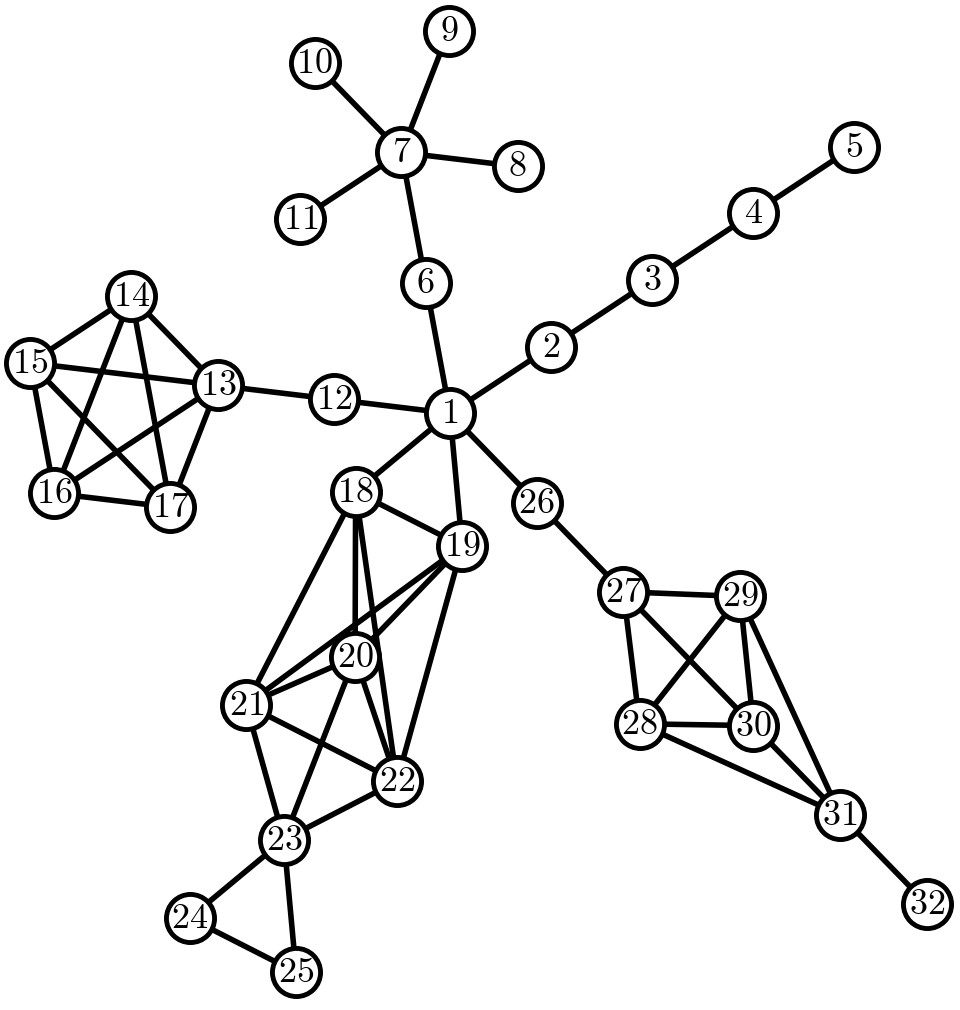}
	\caption{Example topology with five different branches.}
	\label{fig:Example-Branches}
\end{figure}

\section{FMMC Problem over Topologies with Arbitrary Equilibrium Distribution}
\label{sec:FMMCPathBasedResults}

In this section, we address the FMMC problem over topologies with arbitrary (not necessarily symmetric) equilibrium distribution.
Some of the topologies covered in this section are based on the types of subgraphs considered in Section \ref{sec:Branches},
and the
optimal transition probabilities are same as those reported in Section \ref{sec:Branches}.

\subsection{Path Topology}
\label{sec:pathtopology}

The optimal value of $q$ for the path topology are provided in (\ref{eq:Eq201711171267}).
The optimal value of $SLEM$ (i.e. $s$) can be obtained from the recursive solution of  the following set of equations
\begin{subequations}
	\label{eq:Eq201711171300}
	\begin{gather}
		s a_1
		=
		q_1 ( a_{2} / \pi_{2} )
		\label{eq:Eq201711171300a}
		\\
		s a_i
		=
		q_i \left( ( a_{i-1} / \pi_i ) + ( a_{i+1} / \pi_{i+1} ) \right),
		\label{eq:Eq201711171300b}
		\\
		s a_{N-1}
		=
		q_{N-1} ( a_{N-2} / \pi_{N-1} )
		\label{eq:Eq201711171300c}
	\end{gather}
\end{subequations}
where (\ref{eq:Eq201711171300b}) holds for $i=2,...,N-2$.
In the recursive solution of equations (\ref{eq:Eq201711171300}), after deriving $a_{i+1}$ in terms of $a_{i}$ for $i=1, ..., N-1$, a polynomial is obtained in terms of $s$.
The roots of this polynomial are the eigenvalues $\lambda_i$, where $s$ (i.e. $SLEM$) is the largest of these roots.
The optimal value of $q$ reported in (\ref{eq:Eq201711171267}) are in agreement with the results obtained in \cite{Fill2013} and \cite{BoydFastestmixing2003,BoydFastestMixing2009,Boyd2006FMMCPath}, where \cite{Fill2013} addresses the birth and death problem, and their solution is based on stochastically monotone Markov kernels and the monotonicity of mixing times,
and
\cite{BoydFastestmixing2003,BoydFastestMixing2009,Boyd2006FMMCPath} considers the case of uniform equilibrium distribution.
For example,
for $N=2$, the optimal value of $s$ is zero,
for $N=3$, the optimal value of $s$ is
$\sqrt{ \frac{ \pi_{1} \pi_{3} }{ \left( \pi_{1} + \pi_{2} \right) \left( \pi_{2} + \pi_{3} \right) } }$,
for $N=4$, the optimal value of $s$ is
$\sqrt{ \left( ( \pi_1 \pi_3  / ( \pi_1 + \pi_2 ) ) + ( \pi_2 \pi_4 / (\pi_3 + \pi_4 )) \right) / (\pi_2 + \pi_3) }$.

\subsubsection{Path Topology with $N=3$ Vertices}
\label{sec:PathTopologyN3}
If the equilibrium distribution satisfies $\pi_{2}^{2} \geq \pi_{1} \pi_{3}$, the optimal results are as reported above.
If $\pi_{2}^{2} < \pi_{1} \pi_{3}$, the optimal results are as below,
\begin{subequations}
	\label{eq:Eq201808129359-SS}
	\begin{gather}
		q_{1}
		=
		\left( \pi_1 \pi_2 \left( \pi_2 + 2 \pi_3 \right)  \right) / \left(  \pi_2 \left( \pi_1 + \pi_3 \right)  +  4 \pi_1 \pi_3  \right)
		\label{eq:Eq201808129359a-SS}
		\\
		q_{2}
		=
		\left(  \pi_3 \pi_2 \left( \pi_2 + 2 \pi_1 \right)  \right) / \left(  \pi_2 \left( \pi_1 + \pi_3 \right)  +  4 \pi_1 \pi_3  \right)
		\label{eq:Eq201808129359b-SS}
	\end{gather}
\end{subequations}
\begin{equation}
	\label{eq:Eq201808129374-SS}
	\begin{gathered}
		s  =
		\left(  4 \pi_1 \pi_3 - \pi_2^{2}  \right) / \left(  \pi_2 \left( \pi_1 + \pi_3 \right)  +  4 \pi_1 \pi_3  \right).
	\end{gathered}
\end{equation}

\begin{example}
	As an example, we consider a path topology with $N=5$ vertices, with equilibrium distribution $\boldsymbol{\pi} = [ 1.9, 2.9, 3.1, 2.8, 1.7 ]$.
	Using the optimal transition probabilities, the optimal $SLEM$ is equal to $0.748251$, while using the transition probabilities according to Metropolis transition probabilities
	(\ref{eq:Eq20191015-metropolis-transition-probability}),
	the value of $SLEM$ is equal to $0.861111$.
\end{example}

\subsection{Extended Complete Ladder Topology}
\label{sec:ExtendedCompleteLadderTopology}

Here, we consider an extended complete ladder topology where $n$ complete graphs with $m_1, m_2, ..., m_n$ nodes, each connected to each other in the form of a $n$-partite graph.
Regarding the notations, in this subsection, we consider the same notation used in Subsection \ref{sec:ExtendedLadderBranch} for extended complete ladder subgraph.
The optimal transition probabilities for extended complete ladder topology are same as those reported in Subsection \ref{sec:ExtendedLadderBranch} for extended complete ladder subgraph,
and for the transition probabilities between vertices in the $n$-th fiber, we have
$$\boldsymbol{P}(n)_{\alpha_{n}, \beta_{n}} = \frac{  \pi_{n,\beta_{n}}  }{  \sum_{\gamma=1}^{m_{n}} \pi_{n, \gamma}  +  \sum_{\gamma=1}^{m_{n-1}} \pi_{n-1, \gamma}  }$$
for $\alpha_{n} \neq \beta_{n} = 1, ..., m_{n}$.

The base graph of the extended complete ladder topology is a path topology with $n$ vertices.
The corresponding equilibrium distribution of the $i$-th vertex in the base graph is equal to $\sum_{\gamma=1}^{m_{i}} \pi_{i, \gamma}$.
The optimal value of $SLEM$ is equal to the $SLEM$ of the corresponding base graph (i.e. a path graph) which has been addressed in Subsection \ref{sec:pathtopology}.

\begin{example}
	As an example, we consider paw topology with $N=4$ vertices, with equilibrium distribution $\boldsymbol{\pi} = [ 1.9, 3.1, 2.8, 1.7 ]$.
	Paw topology is an extended complete ladder topology with parameters $m_{1}=1$, $m_{2}=2$ and $m_{3}=1$.
	Using the optimal transition probabilities, the optimal $SLEM$ is equal to $0.233425$, while using the transition probabilities according to Metropolis transition probabilities
	(\ref{eq:Eq20191015-metropolis-transition-probability}),
	the value of $SLEM$ is equal to $0.608497$.
\end{example}

\subsection{Extended Barbell Topology}
\label{sec:ExtendedBarbellTopology}

Extended Barbell topology is obtained by connecting two complete graphs $K_{m_1}$, $K_{m_2}$ by a path bridge $P_{n}$.
Extended barbell topology is the clique lift of a path subgraph (with $n+4$ vertices) where the fiber of the vertices $(1)$ and $(n+4)$ is a clique with $m_1-1$ and $m_2-1$ vertices, respectively.
Extended barbell topology is a special case of the extended complete ladder topology (introduced in Subsection \ref{sec:ExtendedCompleteLadderTopology}).
\\
Extended Barbell Topology comprises two Lollipop Subgraphs which are connected to each other from their two ends in the path parts.
Thus the optimal value of $q$ over edges of the Extended Barbell topology is same as those of the Lollipop topology provided in subsection \ref{sec:LollipopBranch}.
\\
For $n=0$ (the case where two complete graphs share a vertex), the optimal value of $SLEM$ is
$$\sqrt{ \frac{(\sum_{\gamma=1}^{m_{1}-1} \pi_{1, \gamma} )(\sum_{\gamma=1}^{m_{2}-1} \pi_{3, \gamma}) }{ \left( \sum_{\gamma=1}^{m_{1}-1} \pi_{1, \gamma} + \pi_{2} \right) \left( \pi_{2} + \sum_{\gamma=1}^{m_{2}-1} \pi_{3, \gamma} \right) } },$$
and for $n=1$ (the case with one edge between two complete graphs), the optimal value of $SLEM$ is
$$\sqrt{ \frac{1}{\pi_2 + \pi_3} \left( \frac{(\sum_{\gamma=1}^{m_{1}-1} \pi_{1, \gamma} ) \pi_3}{\sum_{\gamma=1}^{m_{1}-1} \pi_{1, \gamma}  + \pi_2} + \frac{\pi_2 (\sum_{\gamma=1}^{m_{2}-1} \pi_{4, \gamma} )}{\pi_3 + \sum_{\gamma=1}^{m_{2}-1} \pi_{4, \gamma} } \right) }.$$

\begin{example}
	As an example, we consider an extended barbell topology with parameters $m_{1}=3$, $m_{2}=4$, $n=1$ and equilibrium distribution $\boldsymbol{\pi} = [ 1.9, 1.8, 6.4, 8.1, 2.9, 3.2, 2.1 ]$.
	Using the optimal transition probabilities, the optimal $SLEM$ is equal to $0.653212$, while using the transition probabilities according to Metropolis transition probabilities
	(\ref{eq:Eq20191015-metropolis-transition-probability}),
	the value of $SLEM$ is equal to $0.780862$.
\end{example}

\subsection{Lollipop Topology}
\label{sec:LollipopTopology}

Lollipop topology is a complete graph connected to a path via an edge, each with $m$ and $n$ vertices, respectively.
This
topology is a special case of extended barbel topology (introduced in Subsection \ref{sec:ExtendedBarbellTopology}) where either one of $m_{1}$ or $m_{2}$ is equal to $2$, i.e. one of the fibers at one end of the extended barbell topology is reduced to a single vertex.
The optimal transition probabilities and $SLEM$ can be obtained accordingly from the results provided in Subsection \ref{sec:ExtendedBarbellTopology} by setting either one of $m_{1}$ or $m_{2}$ is equal to $2$.
The results provided for Lollipop topology in this subsection are in agreement with those of \cite{BoydFastestmixing2003}, which has considered the case of Lollipop topology with $m=n=2$ with uniform distribution.

\begin{example}
	As an example, we consider a Lollipop topology with parameters $m=3$, $n=2$ and equilibrium distribution $\boldsymbol{\pi} = [ 0.9, 3.2, 6.5, 3.1, 2.9 ]$.
	Using the optimal transition probabilities, the optimal $SLEM$ is equal to $0.552672$, while using the transition probabilities according to Metropolis transition probabilities
	(\ref{eq:Eq20191015-metropolis-transition-probability}),
	the value of $SLEM$ is equal to $0.610267$.
\end{example}

\subsection{Star Topology}
\label{sec:Star_n1_Topology}
Consider the star topology with $m$ subgraphs of length one.
As explained in Appendix \ref{sec:StarAbstractSolution}, the optimal results for Star topology are as below,
\begin{equation}
	\label{eq:Eq201807237981-SS}
	\begin{gathered}
		q_{i} =
		\left. \left( 2 \pi_{0} \pi_{i} \right) \right/
		\left( 2 \pi_{0} + \Pi \right),
		\;\; \text{for} \;\; i=1,...,m
	\end{gathered}
\end{equation}
\begin{equation}
	\label{eq:Eq201807237991-SS}
	\begin{gathered}
		SLEM
		=
		\Pi /
		\left(  2 \pi_{0} + \Pi  \right),
	\end{gathered}
\end{equation}
if
$\Pi  \leq  2 \pi_{0}$,
where $\pi_{0}$ denotes the equilibrium distribution on the central vertex in the star topology,
and
$\Pi = \sum_{i=1}^{m} \pi_{i}$.
If $\sum_{i=1}^{m} \pi_{i}  >  2 \pi_{0}$, then the optimal results are as below,
\begin{equation}
	\label{eq:Eq201807278235-SS}
	\begin{gathered}
		q_{i}
		=
		\pi_{i} \pi_{0}  \left/  \Pi  \right.
		\;\; \text{for} \;\; i=1,...,m,
	\end{gathered}
\end{equation}
\begin{equation}
	\label{eq:Eq201807278283-SS}
	\begin{gathered}
		SLEM
		=
		\left(  \Pi - \pi_{0}  \right) /  \Pi,
	\end{gathered}
\end{equation}
Note that the results in (\ref{eq:Eq201807237981-SS}), (\ref{eq:Eq201807237991-SS}), (\ref{eq:Eq201807278235-SS}) and (\ref{eq:Eq201807278283-SS}) hold for $m \geq 3$.
In the case of $m=2$, the star topology is reduced to the Path topology with three vertices, where the optimal results are provided in Subsection \ref{sec:PathTopologyN3}.

Windmill and friendship graphs
\cite{Gallian2007FriendshipGraph}
are two well-known graphs which are clique lifted of star graph.
Hence using the optimal results presented in this subsection, the FMRMC problem for all equilibrium distributions can be addressed over both windmill and friendship graphs.

\begin{example}
	As an example, we consider a star topology with $m=4$ branches and equilibrium distribution $\boldsymbol{\pi} = [ 4.9, 2.2, 2.5, 2.1, 1.9 ]$, where the first vertex is the central vertex.
	Using the optimal transition probabilities, the optimal $SLEM$ is equal to $0.47027$, while using the transition probabilities according to Metropolis transition probabilities
	(\ref{eq:Eq20191015-metropolis-transition-probability}),
	the value of $SLEM$ is equal to $0.571792$.
\end{example}

\subsection{Symmetric Star Topology with Semi-Symmetric Equilibrium Distribution}
\label{sec:General-Star-Topology}
Consider the star topology with $m$ subgraphs of length $n$.
Let $\pi_{i, \alpha}$, for $i=1,...,n$, $\alpha = 1,...,m$ be the equilibrium distribution of the vertex on $\alpha$-th subgraph with $i$ hop distance from the central vertex and $\pi_{0}$ be the equilibrium distribution of the central vertex.
We assume that the equilibrium distribution has the property that  $\left( \pi_{i, \alpha} / \pi_{i-1, \alpha} \right) = \chi_{i}$ for  $i=2,...,n$.
If $\chi_{i+1} \geq \chi_{i}$ for $i=2,...,n-1$,
and
$ 2 \pi_{0} \geq \Pi_{1} $,
and
$ 2 \pi_{0} \chi_{2} \leq \Pi_{1} $,
then the optimal value of $q_{i,\alpha}$ is
$\pi_{i,\alpha} \pi_{i-1,\alpha} / \left(  \pi_{i,\alpha} + \pi_{i-1,\alpha} \right)$ for $i=2,...,n$ and $\alpha = 1,...,m$
and
the optimal value of $q_{1,\alpha}$ is
$( 2 \pi_{0} \pi_{1,\alpha} ) / ( 2 \pi_{0} + \Pi_{1} )$,
where
$ \Pi_{1}  =  \sum_{ \alpha = 1 }^{m} \pi_{1,\alpha} $.
For the special case of $n=2$, the optimal value of $SLEM$ is
$(
\beta 	+  \sqrt{
	\beta^2
	+
	8 \chi_{2} /
	(   ( 1 + \chi_{2} )
	( 2 \pi_0 + \Pi_{1} )
	)
}
) / 2 $,
where
$\beta  =  \Pi_{1} / ( 2 \pi_0 + \Pi_{1} )$.
\begin{remark}
	\label{Remark-GenralStarSymmetry}
	An interesting point about the results above is that the optimal transition probabilities between the orbit is equal, where
	the vertices with the same distance from the central vertex form an orbit.
	In other words,
	$\boldsymbol{P}_{ (i,\alpha), (i-1,\alpha) } = 1 / ( 1 + \chi_{i} )$
	and
	$\boldsymbol{P}_{ (i-1,\alpha),(i,\alpha) } = \chi_{i} / ( 1 + \chi_{i} )$ for $i=2,...,n$.
	The transition probabilities of the edges connected to the central vertex follow the same property
	stated in Remark \ref{Remark-StarBranch}.
\end{remark}

\begin{example}
	As an example, we consider a symmetric star topology with $m=3$ branches of length $n=2$ and equilibrium distribution $\boldsymbol{\pi} = [ 4.9, 3.3, 3.6, 2.7, 2.2, 2.4, 1.8 ]$.
	First vertex is the central vertex, and vertices $2,3$ and $4$ are in one hop distance from the central vertex and the rest of vertices are in two hop distance from the central vertex.
	Using the optimal transition probabilities, the optimal $SLEM$ is equal to $0.76053$, while using the transition probabilities according to Metropolis transition probabilities
	(\ref{eq:Eq20191015-metropolis-transition-probability}),
	the value of $SLEM$ is equal to $0.836403$.
\end{example}

\subsection{Bistar Topology}
\label{sec:bistartopology}
Bistar topology comprises two star topology each with $m_{1}$ and $m_{2}$ subgraphs of length one, where their central nodes are connected to each other via one edge.
This topology is depicted in Fig. \ref{fig:DoubleStar}.
\begin{figure}
	\centering
	\includegraphics[width=0.35\hsize]{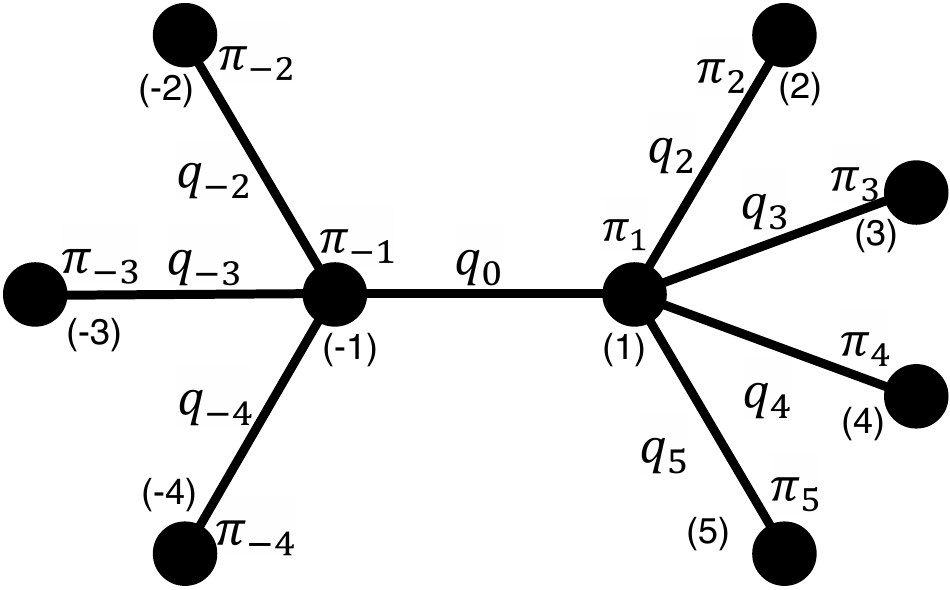}
	\caption{Bistar topology with $m_{1} = 3$, $m_{2} = 4$.}
	\label{fig:DoubleStar}
\end{figure}
The optimal results are
\begin{subequations}
	\label{eq:Eq201807309560-SS}
	\begin{gather}
		q_{i} = \pi_{-1} \pi_{i} / ( \pi_{-1} + \Pi^{-} ),
		\;\; \text{for} \;\; i=-2,...,-m_{1}-1
		\label{eq:Eq201807309560a-SS}
		\\
		q_{i} = \pi_{1} \pi_{i} / ( \pi_{1} + \Pi^{+} ),
		\;\; \text{for} \;\; i=2,...,m_{2}+1
		\label{eq:Eq201807309560b-SS}
	\end{gather}
\end{subequations}
where
$\pi_{-1}$ and $\pi_{1}$ are the equilibrium distributions on the central vertices of the star topologies,
and
$\Pi^{-} = \sum\nolimits_{i=-2}^{-m_{1}-1} \pi_{i}$
and
$\Pi^{+} = \sum\nolimits_{i=2}^{m_{2}+1} \pi_{i}$.
For the optimal value of $q_{0}$ and $s$, we have
\begin{equation}
	\label{eq:Eq201807309579-SS}
	\begin{gathered}
		q_{0} = ( \pi_{-1} \pi_{1} ) / ( \pi_{-1} + \pi_{1} ),
	\end{gathered}
\end{equation}
\begin{equation}
	\label{eq:Eq201807309589-SS}
	\begin{gathered}
		s
		=
		\sqrt{
			\frac{ 1 }{ \pi_{-1} + \pi_{1} }
			\left(
			\frac{ \pi_{1} \Pi^{-} }{ \pi_{-1} + \Pi^{-} }
			+
			\frac{ \pi_{-1} \Pi^{+} }{ \pi_{1} + \Pi^{+} }
			\right)
		}
	\end{gathered}
\end{equation}
if
the equilibrium distribution satisfies the following constraints
\begin{subequations}
	\label{eq:Eq201808029822-SS}
	\begin{gather}
			\Pi^{-}
			\left( \pi_{1} - \Pi^{-} \right)
			\left( \pi_{1} + \Pi^{+} \right)
			+
			\Pi^{+}
			\left( \Pi^{-} + \pi_{-1} \right)^{2}
			> 0,
		\label{eq:Eq201808029822a-SS}
		\\
			\Pi^{+}
			\left( \pi_{-1} - \Pi^{+} \right)
			\left( \pi_{-1} + \Pi^{-} \right)
			+
			\Pi^{-}
			\left( \Pi^{+} + \pi_{1} \right)^{2}
			> 0,
		\label{eq:Eq201808029822b-SS}
	\end{gather}
\end{subequations}
\begin{equation}
	\label{eq:Eq201807309721-SS}
	\begin{gathered}
		\pi_{1} \Pi^{-}    \leq    \pi_{-1}^{2},
		\qquad \text{and} \qquad
		\pi_{-1} \Pi^{+}    \leq    \pi_{1}^{2}.
	\end{gathered}
\end{equation}
Comparing the optimal value of $SLEM$ for double star topology (\ref{eq:Eq201807309589-SS}) and the conditions (\ref{eq:Eq201807309721-SS})
with those of the Path topology with four vertices provided in Subsection (\ref{sec:pathtopology}), it is apparent that the optimal results for the double star topology is similar to that of a Path topology with four vertices where the equilibrium distribution of the two vertices at two ends of the path topology are equal to $\sum_{i=-2}^{-m_{1}-1} \pi_{i}$ and $\sum_{i=2}^{m_{2}+1} \pi_{i}$.
Also, the optimal value of $q$ for the two edges at two ends of the path topology are equal to $\sum_{i=-2}^{-m_{1}-1} q_{i}$ and $\sum_{i=2}^{m_{2}+1} q_{i}$.

\begin{remark}
	\label{remark:bistar}
	The sorted optimal eigenvalues of the transition probability matrix $\boldsymbol{I} - \boldsymbol{D}^{-1} \boldsymbol{L}\left( q \right)$ for the Bistar topology
	are
	$\{  -s, 0, 1-\mu_{2}, 1-\mu_{1}, s, 1  \}$,
	where the eigenvalues $1-\mu_{1}$ and $1-\mu_{2}$ are with degeneracy $m_{1} - 1$ and $m_{2} - 1$.
	On the other hand, consider a Path topology with four vertices
	and
	$[  \sum_{i=-2}^{-m_{1}-1}  \pi_{i}, \pi_{-1}, \pi_{1}, \sum_{i=2}^{m_{2}+1}  \pi_{i}  ]$
	as the equilibrium distribution over its vertices.
	For this Path topology, the sorted optimal eigenvalues of the transition probability matrix
	are
	$ \{  -s, 0, s, 1  \} $
	Comparing these two sequences of the eigenvalues, it is obvious that tight interlacing (as defined in \cite{HAEMERS1995593}) with $k=2$ holds between them.
	Also, the optimal transition probabilities and $SLEM$ provided above are an example conforming Theorem
	\ref{theorem-2-3450}
	and Remark
	\ref{remark-sum-of-eigenvalues},
	i.e. for optimal transition probabilities, both $s$ and $-s$ are eigenvalues of the transition probabilities matrix.
\end{remark}

\begin{example}
	As an example, we consider a Bistar topology with parameters $m_{1}=3$, $m_{2}=4$ and equilibrium distribution $\boldsymbol{\pi} = [ 1.8, 2.3, 1.9, 7.8, 8.3, 2.1, 1.8, 1.7, 2.6 ]$.
	The $4$-th and $5$-th vertices are the central vertices of the star topologies that form the Bistar topology.
	Using the optimal transition probabilities, the optimal $SLEM$ is equal to $0.681843$, while using the transition probabilities according to Metropolis transition probabilities
	(\ref{eq:Eq20191015-metropolis-transition-probability}),
	the value of $SLEM$ is equal to $0.877593$.
\end{example}

\subsection{Complete Graph Topology}
Complete graph is the clique lift of a single vertex.
The optimal $SLEM$ of complete graph is zero and the optimal transition probability matrix for this topology can be written as
$\boldsymbol{P} = \boldsymbol{I} - \boldsymbol{D}^{-1}\boldsymbol{L}(q) = \boldsymbol{I} - \boldsymbol{D}^{-1} \boldsymbol{L}(q) = ( \boldsymbol{J} \boldsymbol{D} ) / ( \sum_{k=1}^{N}\pi_k ) )$.

\section{FMRMC Problem over Topologies with Symmetric Equilibrium Distribution}
\label{sec:FMMCSymmetricResults}
In this section, we address the FMRMC problem over topologies with symmetric equilibrium distribution.

\subsection{Exploiting Symmetry of Graph}
An automorphism of the graph $\mathcal{G} = (\mathcal{V}, \mathcal{E})$ is a permutation $\sigma$ of $\mathcal{V}$ such that $\{i,j\} \in \mathcal{E}$ if and only if $\{\sigma(i),\sigma(j)\}\in \mathcal{E}$. The set of all such permutations, with composition as the group operation, is called the automorphism group of the graph and denoted by $Aut(\mathcal{G})$.
For a vertex $i \in \mathcal{V}$, the set of all images $\sigma(i)$, as $\sigma$ varies through a subgroup $G \subseteq Aut(\mathcal{G})$, is called the orbit of $i$ under the action of $G$.
The vertex set $\mathcal{V}$ can be written as disjoint union of distinct vertex orbits,
where all vertices in each vertex orbit have the same equilibrium probability $(\pi_{i})$.
Similarly the edge set $\mathcal{E}$ can be written as disjoint union of distinct edge orbits.
\begin{proposition}
	\label{AutomorphismProposition}
	The optimal solution of
	the FMRMC problem (\ref{eq:Eq201712231257})
	has the property that
	the value of all weights $q$ over
	edges within an edge orbit are the same
	\cite{BoydFastestMixing2009,SaberThesis2015}.
\end{proposition}

Based on proposition \ref{AutomorphismProposition},
it can be concluded that the optimal weights over all edges of an edge transitive graph are equal.
An edge (vertex) transitive graph is a graph that has only one edge (vertex) orbit.
Note that the edge transitivity and vertex transitivity of a graph are two different properties.

\subsection{Symmetric Tree}
\label{sec:SymmetricTree}
Symmetric tree topology is a tree graph that all vertices on the same depth have the same number of children.
Depth of the tree is denoted by $n$ and the number of children of vertices at $i$-th depth is denoted by $m_{i}$ for $i=0, ..., n-1$, where $i=0$ corresponds to the root vertex and it is assumed that $m_{0} \geq 2$.
The total number of vertices in the symmetric tree topology (denoted by $N$) is equal to $1 + \sum_{i=1}^{n} \prod_{j=0}^{i-1} m_{j}$.
We denote a vertex at $i$-th depth by $(1, \alpha_{1}, ..., \alpha_{i})$, where $\alpha_{j}$ varies from $1$ to $m_{j}$ for $j=1, ..., i$.
The root of the tree is denoted by $(1)$.
A symmetric tree topology with parameters $m_{0} = 2$, $m_{1} = 1$, $m_{2} = 3$ is depicted in Fig. \ref{fig:SymmetricTreeTopology}.
\begin{figure}
	\centering
	\includegraphics[width=120mm]{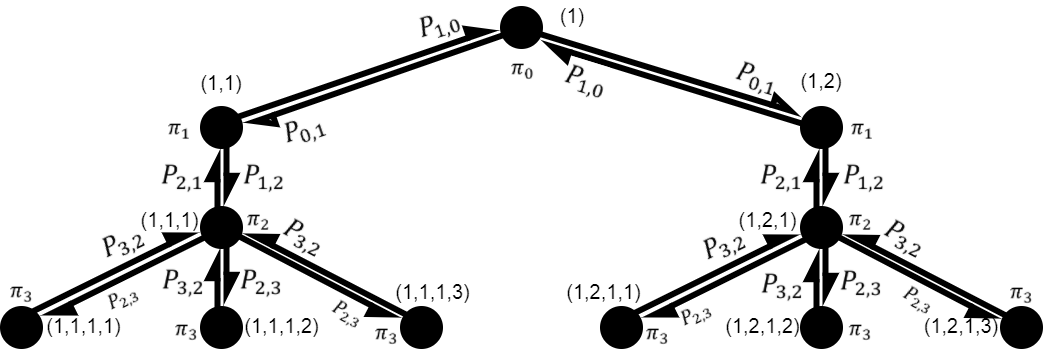}
	\caption{Symmetric tree topology ($m_{0} = 2$, $m_{1} = 1$, $m_{2} = 3$).}
	\label{fig:SymmetricTreeTopology}
\end{figure}

We assume that the equilibrium distribution is symmetric regarding the depth of the tree.
Thus, the automorphism group of the symmetric tree topology is $\mathcal{S}_{m_{0}} \times \mathcal{S}_{m_{1}} \times ... \times \mathcal{S}_{m_{n-1}}$ (permutation group of branches).
In other words, vertices with the same parent and in the same depth are in the same vertex orbit.
As a result, they have the same probability in the equilibrium distribution, i.e. the equilibrium distribution of vertices $( 1, \alpha_{1}, ..., \alpha_{k} )$ is equal to $\pi_{k}$ for $k=1, ..., n$.
Similarly, edges in the same depth are in the same edge orbit, and from Proposition \ref{AutomorphismProposition}, it is obvious that the optimal weight over edges in the same depth is the same.
Thus, for the weights on the edges between vertices $( 1, \alpha_{1}, ..., \alpha_{k} )$ and $( 1, \alpha_{1}, ..., \alpha_{k}, \alpha_{k+1} )$ we use $q_{k}$ for $k = 1, ..., n-2$
and for the weights on the edges between vertices $(1)$ and $( 1, \alpha_{1} )$, we use $q_{0}$.
As explained in Appendix \ref{sec:SymmetricTree-Appendix},
For the equilibrium distributions that satisfy
$ 2 \pi_{0}\geq m_{0} \pi_{1} $,
$m_{0} \pi_{1}^{2} \geq 2 m_{1} \pi_{0} \pi_{2}$
and
$m_{k} \pi_{k+1}^{2} \geq m_{k+1} \pi_{k} \pi_{k+2}$
for $k = 1, ..., n-2$,
the optimal value of $q$ is as below,
\begin{equation}
	\label{eq:Eq201801186054}
	\begin{gathered}
		q_{0}  =  2 \pi_{0} \pi_{1} / ( 2 \pi_{0}  +  m_{0} \pi_{1} ),
	\end{gathered}
\end{equation}
\begin{equation}
	\label{eq:Eq201801186064}
	\begin{gathered}
		q_{k}  =  \pi_{k} \pi_{k+1}  /  ( \pi_{k} + m_{k} \pi_{k+1} )
		\;\; \text{for} \;\; k = 1,...,n-1,
	\end{gathered}
\end{equation}
and
the optimal value of $SLEM$ (i.e. $s$) can be obtained from the recursive solution of  the following set of equations
\begin{subequations}
	\label{eq:Eq201801186075}
	\begin{gather}
			s  a_{0}	=   ( \sqrt{m_{1}} /  \pi_{1} )  a_{1},
			\;\;
			s  a_{n-1}    =    ( \sqrt{m_{n-1}} / \pi_{n-1} )  a_{n-2}
		\label{eq:Eq201801186075a}
		\\
		s  a_{k}	=   ( \sqrt{m_{k}} / \pi_{k} )  a_{k-1}    +    ( \sqrt{m_{k+1}} / \pi_{k+1} )  a_{k+1}
		\label{eq:Eq201801186075b}
	\end{gather}
\end{subequations}
where (\ref{eq:Eq201801186075b}) holds for $k=1, ..., n-2$.
In the recursive solution of equations (\ref{eq:Eq201801186075}), after deriving $a_{k+1}$ in terms of $a_{k}$ for $i=1, ..., n-2$, a polynomial is obtained in terms of $s$.
The roots of this polynomial are the eigenvalues $\lambda_i$, where $s$ (i.e. $SLEM$) is the largest of these roots.

\subsubsection{Symmetric Tree of Depth Two}
\label{sec:SymmetricTreeDepth2}
In the following, as an example we provide the optimal $SLEM$ for a symmetric tree topology of depth two with all possible equilibrium distributions.
More detailed solution is provided in Appendix \ref{sec:SymmetricTreeDepth2-Appendix}.
\subsubsection*{Case $1$}
For equilibrium distributions that satisfy
$ 2 \pi_{0}\geq m_{0} \pi_{1} $
and
$  m_{0} \pi_{1}^{2} \geq 2 m_{1} \pi_{0} \pi_{2}$,
the optimal value of $q$ is same as those provided in (\ref{eq:Eq201801186054}) and (\ref{eq:Eq201801186064})
and
the optimal value of $SLEM$ is
$\frac{    \sqrt{  \frac{  A  }{  m_{1}\pi_{2} + \pi_{1}  }  }  +  m_0\pi_1  }{  2( m_{0} \pi_{1} + 2 \pi_{0})   }$,
where
$A = m_{0}^{2} \pi_{1}^{2} \left( m_{1} \pi_{2} + \pi_{1} \right) + 8 m_{1}\pi_0\pi_2 (m_{0}\pi_{1} + 2\pi_{0})$.
\subsubsection*{Case $2$}
For equilibrium distributions that satisfy
$2 \pi_{0}< m_{0} \pi_{1}$
and
$m_{0} \pi_{1}^{2} \geq 2 m_{1} \pi_{0} \pi_{2}$,
the optimal values of $q_0$ and $q_1$ are $q_0=\frac{\pi_0}{m_0}$ and $q_1=\frac{(3m_0\pi_1-2\pi_0)\pi_2}{2m_0(m_1\pi_2+\pi_1)}$, respectively,
and
the optimal value of $SLEM$ is
$\left. \sqrt{  \frac{  A  }{  m_{1} \pi_{2} + \pi_{1}  }  }  \right/
(  4 \pi_{1} m_{0}  )
+
\frac{  7  }{  4  } $
where
$A = \pi_{1} ( 9 m_{0}^{2} \pi_{1} ( m_{1} \pi_{2} + \pi_{1} ) - 24 m_{0} \pi_{0} \pi_{1} + 16 \pi_{0}^{2} )$.
\subsubsection*{Case $3$}
For equilibrium distributions that satisfy
$ 2 \pi_{0}\geqslant m_{0} \pi_{1}$
and
$m_{0} \pi_{1}^{2} < 2 m_{1} \pi_{0} \pi_{2}$,
the optimal values of $q_0$ and $q_1$ are $q_0= \frac{2 \pi_0 (m_1 \pi_2 - \pi_1)}{ m_0 m_1 \pi_2 - 2 \pi_0}$ and $q_1=\frac{ \pi_2 (m_0 \pi_1 - 2 \pi_0)
}{ m_0·m_1 \pi_2 - 2 \pi_0}$, respectively,
and
the optimal $SLEM$ is
$\frac{\sqrt{A}+\pi_1(m_0(3m_1\pi_2 + \pi_1) - 8\pi_0)}{2\pi_1(m_0m_1\pi_2 - 2\pi_0)} $ where
$A$ $=$
$\pi_1$ $($ $m_0^2 \pi_1$ $($ $m_1^2 \pi_2^2$  $+$ $2 m_1 \pi_1 \pi_2$ $+$ $\pi_1$ $)$ $-$ $16 m_0 m_1 \pi_0 \pi_1 \pi_2$ $+$ $16 m_1 \pi_0^2 \pi_2$ $)$ $)$.
\subsubsection*{Case $4$}
For equilibrium distributions that satisfy
$m_{0} \pi_{1} > 2 \pi_{0}$
and
$2 m_{1} \pi_{0} \pi_{2}> m_{0} \pi_{1}^{2}  $,
the optimal value of $q_0$ and $q_1$ are
$q_0=\frac{\pi_0}{m_0}$ and $q_1=\frac{m_0\pi_1-\pi_0}{m_0m_1}$, respectively,
and
the optimal value of $SLEM$ is
$\frac{  \sqrt{  A  }  }{  2 \pi_{1} \pi_{2} m_{0} m_{1}  }   +   \frac{  m_0 (m_1·\pi_2 - \pi_1) + \pi_0  }{  2 m_{0} m_{1} \pi_{2}  }$
where
$A$ $=$ $\pi_{1}$ $($ $m_{0}^{2} \pi_{1}$ $($ $m_{1}^{2} \pi_{2}^{2}$ $+$ $2 m_{1} \pi_{1} \pi_{2}$ $+$ $\pi_{1}^{2}$ $)$ $-$ $2 m_{0} \pi_{0} \pi_{1}$ $($ $3 m_{1} \pi_{2}$ $+$ $\pi_{1}$ $)$ $+$ $\pi_{0}^{2}$ $($ $4 m_{1} \pi_{2}$ $+$ $\pi_{1}$ $)$ $)$ .

\subsubsection{Special Case of Symmetric Tree Topology}
\label{sec:Special-Case-of-Symmetric-Tree-Topology}
Considering a symmetric tree of depth $n$,
if the equilibrium distribution satisfies
$m_{0} \pi_{1} > 2 \pi_{0}$
and
$2 m_{i} \pi_{i-1} \pi_{i+1}> m_{i-1} \pi_{i}^{2}$ for
$i=1,...,k$,
it can be shown that the
optimal $q_{i}$ for
$i=0,...,k$
is
$q_{i}$ $=$ $($  $m_{0}$ $($ $m_{1} ...$ $($ $m_{i-2}$ $($ $m_{i-1} \pi_{i}$ $-$ $\pi_{i-1}$ $)$ $+$ $\pi_{i-1}$ $)$ $...$ $+$ $(-1)^{i-1}\pi_{1}$ $)$ $+$ $(-1)^{i}\pi_{0}$  $) / ($   $\prod_{j=0}^{i} \pi_{j}$  $)$,
where $k \leq n-1$.

\begin{example}
	As an example, we consider a symmetric tree topology with parameters $n=3$, $m_{0}=2$, $m_{1}=1$, $m_{2}=3$ and equilibrium distribution $\boldsymbol{\pi}_{0} = 6.1$, $\boldsymbol{\pi}_{1} = 5.5$, $\boldsymbol{\pi}_{2} = 3.4$, $\boldsymbol{\pi}_{3} = 0.3$.
	Using the optimal transition probabilities, the optimal $SLEM$ is equal to $0.793041$, while using the transition probabilities according to Metropolis transition probabilities
	(\ref{eq:Eq20191015-metropolis-transition-probability}),
	the value of $SLEM$ is equal to $0.865453$.
\end{example}

\subsection{Symmetric Star}
\label{sec:SymmetricStar}

Symmetric star topology of order $(m,n)$ consists of $m$ path subgraphs of length $n$, connected to one central node.
We assume that the equilibrium distribution is symmetric regarding the distance from the central node, i.e.,
the equilibrium distribution for vertices $(i,j)$ is equal to $\pi_{j}$ for $i=1, ..., m$, $j=1, ..., n$ and the equilibrium distribution for the central node is denoted by $\pi_{0}$.

Symmetric star topology is a special case of symmetric tree topology with $m_{0} = m$ and $m_{i}=1$ for $i = 1, ..., n$.
Based on the results presented in Subsection \ref{sec:SymmetricTree},
we can conclude that the optimal $q_{j}$ is equal to $\frac{\pi_{j} \pi_{j+1} } {\pi_{j} + \pi_{j+1}}$ for $j=1, ..., n-1$
and
the optimal value of $q_{0}$ is $( 2 \pi_{1} \pi_{0} ) / ( m \pi_{1} + 2 \pi_{0} )$, given that
$m \pi_{1} \leq 2 \pi_{0}$,
$2 \pi_{0} \pi_{2} \leq m \pi_{1}^{2}$,
and
$\pi_{i} \pi_{i+2} \leq \pi_{i+1}^{2}$, for $i = 1, 2, ..., n-2$.
Also, the symmetric star topology discussed here is a special case of the star topology addressed in Subsection \ref{sec:General-Star-Topology}, where $\pi_{i,\alpha} = \pi_{i}$ for $i=1,...,n$, and $\alpha = 1, ..., m$.

For $n=1$, the optimal value of $SLEM$ is equal to
$\frac{ m \pi_{1} }{ m \pi_{1} + 2 \pi_{0} }$.
In \cite{Cihan2015}, the special case of a star topology with $n=1$ and the equilibrium distribution $\pi_{0} = m / ( m+1 )$, $\pi_{1} = 1 / ( m+1 )$ has been studied where the optimal $SLEM$ is equal to $1/3$ and the optimal transition probabilities are $\boldsymbol{P}_{0,1} = 2 / ( 3 m )$ and $\boldsymbol{P}_{1,0} = 2 / 3$.
The results in \cite{Cihan2015}
are in agreement with the results obtained in this subsection.
For $n=2$, the optimal value of $SLEM$ is
$\frac{ m \pi_{1} }{ 2 \left( m \pi_{1} + 2 \pi_{0} \right) }   +   \frac{1}{2} \sqrt{  \frac{ m^2 \pi_{1}^{2} }{ \left( 2\pi_{0} + m \pi_{1} \right)^{2} }  +  \frac{ 8 \pi_{0} \pi_{2} }{ \left( \pi_{1} + \pi_{2} \right) \left( 2\pi_{0} + m \pi_{1} \right) }  }$.

\begin{example}
	As an example, we consider a symmetric star topology with parameters $n=2$, $m=3$ and equilibrium distribution $\boldsymbol{\pi}_{0} = 5.3$, $\boldsymbol{\pi}_{1} = 3.1$,  $\boldsymbol{\pi}_{2} = 1.9$.
	Using the optimal transition probabilities, the optimal $SLEM$ is equal to $0.740632$, while using the transition probabilities according to Metropolis transition probabilities
	(\ref{eq:Eq20191015-metropolis-transition-probability}),
	the value of $SLEM$ is equal to $0.84752$.
\end{example}

\subsection{Complete Cored Symmetric (CCS) Star}
\label{sec:CCSStar}

CCS star topology of order $(m,n)$ consists of $m$ path subgraphs of length $n$, connected to each other at one end such that the connected part of path subgraphs form a complete graph, called core.
Each one of path subgraphs contains $n$ edges.
A CCS star graph has $|\mathcal{V}| = m(n+1)$ nodes and $|\mathcal{E}| = nm + m(m+1)/2$ edges.
We denote the nodes of the graph by $\mathcal{V}  =  \{ (i,j) | i=1, ..., m, \;\; j=0, ..., n \}$.

We assume that the equilibrium distribution is symmetric regarding the distance from the central core.
Thus, the automorphism group of CCS Star topology is $\mathcal{S}_{m}$ (permutation group of subgraphs).
In other words, nodes with the same distance from the central core are in the same vertex orbit.
As a result, they have the same probability in the equilibrium distribution, i.e. the equilibrium distribution for nodes $(i,j)$ is equal to $\pi_{j}$ for $i=1, ..., m$, $j=0, ..., n$.
Similarly, edges with the same distance from the central core are in the same edge orbit, and from Proposition \ref{AutomorphismProposition}, it is obvious that the optimal weight over edges with the same
distance from the central core is the same.
Thus, for the weights on the edge between nodes $(i,j-1)$ and $(i,j)$ (for $i=1, ..., m$, $j=1, ..., n$), we use $q_{j}$ for $j=1, ..., n$ and for the weights on the edges in the central core, we use $q_{0}$.

The optimal
results
are
$q_{j} = \frac{ \pi_{j} \pi_{j-1} }  { \pi_{j} + \pi_{j-1} }$ for $j=1, ..., n$
and $q_{0} = \pi_{0} / m$
given that $(m-1) \pi_{1} \leq \pi_{0}$ and the constraint (\ref{eq:Eq201711171284}) are satisfied.
For $n=1$, the optimal value of $SLEM$ is equal to
$\sqrt{  \frac{ \pi_{1} }{ \pi_{0} + \pi_{1} }  }$,
and
for $n=2$, the optimal value of $SLEM$ is equal to
$\sqrt{  \frac{ \pi_{1} }{ \pi_{0} + \pi_{1} }  +  \frac{ \pi_{0} \pi_{1} }{ \left( \pi_{0} + \pi_{1} \right) \left( \pi_{1} + \pi_{2} \right) }  }$.
In the case of $n=2$, if $(m-1) \pi_{1} > \pi_{0}$, then the optimal results are
$q_{0} = (  \pi_{0} ( 2 \pi_{1} + \pi_{0} )  ) / (  (m-1) \pi_{0} + ( 4 m - 2 )\pi_{1}  )$,
$q_{1} = (  2 m \pi_{0} \pi_{1}  ) / (  (m-1)\pi_{0} + ( 4 m - 2 )\pi_{1}  )$,
and
the optimal $SLEM$ is
$(  2 ( 2 m - 1 ) \pi_{1} - \pi_{0}  ) / (  ( m - 1 ) \pi_{0} + ( 4 m - 2 )\pi_{1}  )$.
The detailed solution for this topology is provided in Appendix \ref{sec:CCSStar-Appendix}.

Considering the definition of lifted graph (Definition \ref{LiftDefinition}), it is obvious that the CCS star topology is the path lift of a complete graph,
and interestingly, the FMRMC problem over the lifted graph is reduced to the FMRMC problem over the fiber graph (a path graph) with a self loop in the vertex connected to the base.

\begin{example}
	As an example, we consider a CCS star topology with parameters $n=2$, $m=3$ and equilibrium distribution $\boldsymbol{\pi}_{0} = 12.8$, $\boldsymbol{\pi}_{1} = 4.7$,  $\boldsymbol{\pi}_{2} = 1.1$.
	Using the optimal transition probabilities, the optimal $SLEM$ is equal to $0.638193$, while using the transition probabilities according to Metropolis transition probabilities
	(\ref{eq:Eq20191015-metropolis-transition-probability}),
	the value of $SLEM$ is equal to $0.702632$.
\end{example}

\section{Conclusions}
\label{sec:Conclusions}

This paper addresses optimization of the reversible Markov chains in terms of their mixing rate towards the given equilibrium distribution, by providing the optimal transition probabilities.
Two general scenarios for the underlying topologies of the FMRMC problem have been considered.
First scenario is a subgraph connected to an arbitrary graph where we have shown that the optimal transition probabilities over the edges of the subgraph can be determined independent of the rest of the underlying topology.
Second scenario is the clique lifted graphs where we have shown that the FMRMC problem over given clique lifted graph is reducible to the FMRMC problem over its base graph, and the optimal convergence rate of the clique lifted graph is equal to that of its base graph.
Based on these results, we  have addressed the FMRMC problem over a
variety
of stand-alone topologies and subgraphs with common base graph.
Among the topologies addressed in this paper, star topology is of particular interest.
Since,
the optimal results of the FMRMC problem have been provided for all possible equilibrium distributions over the star topology,
and also,
it serves as the base topology for a number of well-known topologies including friendship and windmill topologies.
We have shown the gain in terms of convergence rate, by comparing the optimal transition probabilities with Metropolis transition probabilities.
This comparison has been performed for every topology and subgraph that the optimal results have been provided.
In the case of CCS star topology we have shown that the FMRMC problem over the lifted graph is reduced to the FMRMC problem over the fiber graph rather than the base graph.
This is an interesting feature of the lifted graphs which is worthy of further investigation.
Also, another interesting topic within the context of lifted graphs is other types of lift, i.e. lifts other than clique lift.

\appendix

\section{Cauchy Interlacing Theorem \cite{Golub1989}}
\label{sec:CauchyInterlacingTheorem}

Let $\boldsymbol{A}$ and $\boldsymbol{B}$ be $n \times n$ and $m \times m$ matrices, where $m \leq n$, $\boldsymbol{B}$ is called a compression of $\boldsymbol{A}$ if there exists an
projection $\boldsymbol{R}$ onto a subspace of dimension $m$ such that $\boldsymbol{R} \times \boldsymbol{A} \times \boldsymbol{R}^{T} = \boldsymbol{B}$ and $\boldsymbol{R} \times \boldsymbol{R}^{T} = \boldsymbol{I}_{m}$.
The Cauchy interlacing theorem states that the eigenvalues of $\boldsymbol{B}$ interlace the eigenvalues of $\boldsymbol{A}$,
i.e. if the eigenvalues of
$\boldsymbol{A}$ are
$\lambda_{n}\left( \boldsymbol{A} \right) \leq ... \leq \lambda_{1}\left( \boldsymbol{A} \right)$,
and those of $\boldsymbol{B}$ are
$\lambda_{m}\left( \boldsymbol{B} \right) \leq ... \leq \lambda_{1}\left( \boldsymbol{B} \right)$,
then for all $j=1,...,m$, we have
\begin{equation}
	\nonumber
	\begin{gathered}
		\lambda_{n-j+1}\left( \boldsymbol{A} \right)
		\leq
		\lambda_{m-j+1}\left( \boldsymbol{B} \right)
		\leq
		\lambda_{m-j+1}\left( \boldsymbol{A} \right).
	\end{gathered}
\end{equation}

\section{Path Subgraph}
\label{sec:pathbranch-Appendix}
Substituting $\boldsymbol{L}_{B}\left( q \right)$ (given in (\ref{eq:Eq201711181364})) in (\ref{eq:Eq201711251504}),
equations (\ref{eq:Eq201712021281}) and (\ref{eq:Eq201712021292}) for the path subgraph can be written as
${\boldsymbol{Z}_1}_{B}    =    \sum\nolimits_{i=1}^{N-1} a_{i} \boldsymbol{D}_{B}^{-\frac{1}{2}} \left( {\boldsymbol{e}_{B}}_{i} - {\boldsymbol{e}_{B}}_{i+1} \right)$
and
${\boldsymbol{Z}_2}_{B}$    $=$    $\sum\nolimits_{i=1}^{N-1} b_i \boldsymbol{D}_{B}^{-\frac{1}{2}} \left( {\boldsymbol{e}_{B}}_{i} - {\boldsymbol{e}_{B}}_{i+1} \right)$,
where by substituting them in (\ref{eq:Eq201711251504}),
we have
\begin{subequations}
	\label{eq:Eq201711171119}
	\begin{gather}
			\left( s - 1 \right) a_{1}  +  q_{1} \left(a_1
			\widehat{\pi}_1
			-
			( a_{2} / \pi_{2} )
			\right)  =  0
		\label{eq:Eq201711171119a}
		\\
			\left( s - 1 \right) a_{i}  +  q_{i} \left(a_i
			\widehat{\pi}_i
			-
			( a_{i-1} / \pi_i )
			-
			( a_{i+1} / \pi_{i+1} )
			\right)  =  0,
		\label{eq:Eq201711171119b}
		\\
			\left( s - 1 \right) a_{N-1}  +  q_{N-1} \left(a_{N-1}
			\widehat{\pi}_{N-1}
			-
			\frac{ a_{N-2} }{ \pi_{N-1} }
			-
			\frac{a_{N}}{\pi_{N}}
			\right)  =  0
		\label{eq:Eq201711171119c}
		\\
			\left( s + 1 \right) b_{1}  -  q_{1} \left(b_1
			\widehat{\pi}_{1}
			-
			( b_{2} / \pi_{2} )
			\right)  =  0
		\label{eq:Eq201711171119d}
		\\
			\left( s + 1 \right) b_{i}  -  q_{i} \left(b_i
			\widehat{\pi}_{i}
			-
			( b_{i-1} / \pi_i )
			-
			( b_{i+1} / \pi_{i+1} )
			\right)  =  0,
		\label{eq:Eq201711171119e}
		\\
			\left( s + 1 \right) b_{N-1}  -  q_{N-1} \left(b_{N-1}
			\widehat{\pi}_{N-1}
			-
			\frac{b_{N-2}}{\pi_{N-1}}
			-
			\frac{b_{N}}{\pi_{N}}
			\right)  =  0
		\label{eq:Eq201711171119f}
	\end{gather}
\end{subequations}
where
$\widehat{\pi}_i = ( 1 / \pi_i ) + ( 1 / \pi_{i+1} )$
and
(\ref{eq:Eq201711171119b}) and (\ref{eq:Eq201711171119e}) hold for $i=2,...,N-2$.
Since the edges in path subgraph are not part of any loop, then from (\ref{eq:Eq20171114977}), we have,
\begin{equation}
	\label{eq:Eq201711171166}
	\begin{gathered}
		\left( (s-1) a_{i} \right)^{2}  =  \left( (s+1) b_{i} \right)^{2},
	\end{gathered}
\end{equation}
for $i=1,...,N-1$.
Relations in (\ref{eq:Eq201711171166}) can also be obtained from (\ref{eq:Eq201712011530}).
From (\ref{eq:Eq201711171166}),
we can conclude the following
\begin{equation}
	\label{eq:Eq201711171181}
	\begin{gathered}
		\left(  a_{i} / a_{i+1} \right)^{2}  =   \left(  b_{i} / b_{i+1} \right)^{2},
		\;\; \text{for} \;\; i=1,...,N-2.
	\end{gathered}
\end{equation}
On the other hand from (\ref{eq:Eq201711171119a}) and (\ref{eq:Eq201711171119d}), we have
$( s -1+   q_{1} ( ( 1 / \pi_1 ) + ( 1 / \pi_{2} ) ) )^2 a_1^2 = ( q_1 ( a_{2} / \pi_{2} ) )^2$
and
$( s +1-   q_{1} ( ( 1 / \pi_1 ) + ( 1 / \pi_{2} ) ) )^2 b_1^2 = ( q_1 ( b_{2} / \pi_{2} ) )^2$,
where considering (\ref{eq:Eq201711171181}) for $i=1$,
it can be concluded that
		$( s -1+   q_{1} ( ( 1 / \pi_1 ) + ( 1 / \pi_{2} ) ) )^2
		=
		( s +1-   q_{1} ( ( 1 / \pi_1 ) + ( 1 / \pi_{2} ) ) )^2$
For this equation to be satisfied,
either $s$ should be zero (which is not acceptable) or the following should hold
$q_1 = ( \pi_1\pi_2 ) / ( \pi_1+\pi_2 )$.
Substituting $q_1$
in (\ref{eq:Eq201711171119a}) and (\ref{eq:Eq201711171119d}), $a_1$ ($b_1$) can be written in terms of $a_2$ ($b_2$).
Using these results in (\ref{eq:Eq201711171119b}) and (\ref{eq:Eq201711171119e}) for $i=2$,
and considering (\ref{eq:Eq201711171181}) for $i=2$,
we have
		$( s -\frac{q_1q_2}{(s\pi_2)^2}-1  +  q_{2} (\frac{1}{\pi_2} +\frac{1}{\pi_{3}}) )^2
		=
		( s -\frac{q_1q_2}{(s\pi_2)^2}+1  -  q_{2} (\frac{1}{\pi_2} +\frac{1}{\pi_{3}}) )^2$,
where the following can be concluded
$q_2 = ( \pi_2 \pi_3 ) / ( \pi_2+\pi_3 )$.
Continuing this procedure recursively,
(\ref{eq:Eq201711171267}) is obtained for the optimal value of $q_i$.
In order to satisfy the last constraints of (\ref{eq:Eq201801262501}),
the equilibrium distribution
has to satisfy
(\ref{eq:Eq201711171284}).

\section{Star Topology}
\label{sec:StarAbstractSolution}
Consider the star topology with $m$ subgraphs of length one.
Let
$\mathcal{V} = \{(i)| i=1,...,m\}\cup\{(0)\}$
be the set of nodes (with $(0)$ denoting the central node), and
$\mathcal{E} = \{(0,i)| i=1,...,m\}$ be the set of edges.
Considering
$\boldsymbol{e}_{i}$ for $i\in\mathcal{V}$ as the column vector with $i$-th element equal to one and zero otherwise,
matrix $\widetilde{\boldsymbol{L}}\left( q \right)$ can be written as
$\sum_{i=1}^{m} q_{i} \boldsymbol{e}_{0i} \boldsymbol{e}_{0i}^{T}$,
where
$\boldsymbol{e}_{0i} = \frac{\boldsymbol{e}_{0}}{\sqrt{\pi_0}} - \frac{\boldsymbol{e}_{i}}{\sqrt{\pi_i}}$
for
$i=1,...,m$.
Since the star topology is a tree graph the vectors
$\{ \boldsymbol{e}_{0i} | i=1,...,m \}$
are independent of each other and thus, they can form a basis.
Using
these basis,
(\ref{eq:Eq201801262582}) can be written as
$\boldsymbol{Z}_{1} = \sum\nolimits_{i=1}^{m} a_{i} \boldsymbol{e}_{0i}$,
$\boldsymbol{Z}_{2} = \sum\nolimits_{i=1}^{m} b_{i} \boldsymbol{e}_{0i}$,
$\boldsymbol{Z}_{3} = \sum\nolimits_{ i \in \mathcal{V} } c_{i} e_{i}$
and (\ref{eq:Eq20171114977})
as below,
\begin{equation}
	\label{eq:Eq201807237745-SS}
	\begin{gathered}
		\left( (s-1) a_{i} \right)^{2}  =  \left( (s+1) b_{i} \right)^{2},
		\;\;\quad \text{for} \;\; i=1,...,m.
	\end{gathered}
\end{equation}
Considering the Gram matrix
$\boldsymbol{G}_{ii} = \boldsymbol{e}_{0i}^{T}  \boldsymbol{e}_{0i} = \frac{1}{\pi_{0}} + \frac{1}{\pi_{i}}$, for $i=1,...,m$,
and
$\boldsymbol{G}_{ij} = \boldsymbol{e}_{0i}^{T}  \boldsymbol{e}_{0j} = \frac{1}{\pi_{0}}$ for $i \neq j=1,...,m$,
(\ref{eq:Eq20171114964}) can be written as
\begin{subequations}
	\label{eq:Eq201807237825-SS}
	\begin{gather}
		( a_{i} q_{i} / \pi_{i} ) + ( q_{i} / \pi_{0} ) \sum\nolimits_{j=1}^{m}  a_{j}
		=
		( 1 - s ) a_{i},
		\label{eq:Eq201807237825a-SS}
		\\
		( b_{i} q_{i} / \pi_{i} ) + ( q_{i} / \pi_{0} ) \sum\nolimits_{j=1}^{m}  b_{j}
		=
		( 1 + s ) b_{i},
		\label{eq:Eq201807237825b-SS}
	\end{gather}
\end{subequations}
for $i=1,...,m$,
Since
$\frac{1}{\pi_{0}} (\sum\nolimits_{j=1}^{m}  a_{j}) $ and $\frac{1}{\pi_{0}} (\sum\nolimits_{j=1}^{m} b_{j}) $
are independent of $i$, thus
following can be concluded from (\ref{eq:Eq201807237825-SS})
\begin{subequations}
	\label{eq:Eq201807237841-SS}
	\begin{gather}
			\left( \left( q_{i} / \pi_{i} \right) + s - 1 \right)  ( a_{i} / \pi_i )
			=
			\left( \left( q_{j} / \pi_{j} \right) + s - 1 \right)  ( a_{j} / \pi_j ),
		\label{eq:Eq201807237841a-SS}
		\\
			\left( \left( q_{i} / \pi_{i} \right) - s - 1 \right)  ( b_{i} / \pi_i )
			=
			\left( \left( q_{j} / \pi_{j} \right) - s - 1 \right)  ( b_{j} / \pi_j ),
		\label{eq:Eq201807237841b-SS}
	\end{gather}
\end{subequations}
for $i=1,...,m$,
Considering (\ref{eq:Eq201807237745-SS}), from (\ref{eq:Eq201807237841-SS}) we have
\begin{equation}
	\label{eq:Eq201807237861-SS}
	\begin{gathered}
		\left(
		\frac
		{ \left( q_{i} / \pi_{i} \right) + s - 1 }
		{ \left( q_{i} / \pi_{i} \right) - s - 1 }
		\right)^{2}
		=
		\left(
		\frac
		{ \left( q_{j} / \pi_{j} \right) + s - 1 }
		{ \left( q_{j} / \pi_{j} \right) - s - 1 }
		\right)^{2},
	\end{gathered}
\end{equation}
for $i,j=1,...,m$.
where the only acceptable conclusion is
$q_{i} / \pi_{i}    =    q_{j} / \pi_{j}$
for $i,j=1,...,m$.
Defining $\mu = \frac{ q_{i} }{ \pi_{i} }$, equations (\ref{eq:Eq201807237825-SS}) can be written as below,
\begin{subequations}
	\label{eq:Eq201807237897-SS}
	\begin{gather}
		a_{i} \mu + ( \mu \pi_{i}  /  \pi_{0} ) \sum\nolimits_{j=1}^{m}  a_{j}  =  ( 1 - s ) a_{i},
		\label{eq:Eq201807237897a-SS}
		\\
		b_{i} \mu + ( \mu \pi_{i} / \pi_{0} ) \sum\nolimits_{j=1}^{m}  b_{j}  =  ( 1 + s ) b_{i},
		\label{eq:Eq201807237897b-SS}
	\end{gather}
\end{subequations}
for $i=1,...,m$.
summing over $i=1,...,m$, we have
\begin{subequations}
	\label{eq:Eq201807237929-SS}
	\begin{gather}
		\left(
		\mu \left( 1 + ( 1 / \pi_{0} ) \sum\nolimits_{i=1}^{m} \pi_{i}  \right)  -  \left( 1 - s \right)
		\right)
		\sum\nolimits_{i=1}^{m} a_{i}
		=
		0,
		\label{eq:Eq201807237929a-SS}
		\\
		\left(
		\mu \left( 1 + ( 1 / \pi_{0} )  \sum\nolimits_{i=1}^{m} \pi_{i}  \right)  -  \left( 1 + s \right)
		\right)
		\sum\nolimits_{i=1}^{m} b_{i}
		=
		0,
		\label{eq:Eq201807237929b-SS}
	\end{gather}
\end{subequations}
where the acceptable answer is obtained for $\sum_{i=1}^{m} a_{i} = 0$ and $\sum_{i=1}^{m} b_{i}  \neq 0$, as
reported in
(\ref{eq:Eq201807237981-SS})
and
(\ref{eq:Eq201807237991-SS}).
From (\ref{eq:Eq201807237991-SS}), it is obvious that for $\sum_{i=1}^{m} \pi_{i} \leq 2\pi_{0}$, we have $SLEM \leq \frac{1}{2}$.

The results in
(\ref{eq:Eq201807237981-SS}) and (\ref{eq:Eq201807237991-SS}) hold true if
(\ref{eq:Eq20171114955}) holds true, i.e., $\pi_{0} \geq \sum_{i=1}^{m} q_{i}$.
Considering (\ref{eq:Eq201807237981-SS}), this constraint can be written as $\sum_{i=1}^{m} \pi_{i} \leq 2 \pi_{0}$.
Therefore, for $\sum_{i=1}^{m} \pi_{i} > 2 \pi_{0}$,
and assuming $b_{i} = 0$ for $i=1,...,m$,
(\ref{eq:Eq20171114916}) can be written as
$\left( \left( s - 1 \right) \frac{a_{i}}{q_i} \right)^{2} = c_{0}^{2} + c_{i}^{2}$, for $i=1,...,m$,
and from the fact that $c_{i}=0$ for $i=1,...,m$, it can be concluded that $\left( \left( s - 1 \right) \frac{a_{i}}{q_i} \right)^{2} = c_{0}^{2}$, and thus $(\frac{|a_{i}|}{q_i})^{2} = (\frac{|a_{j}|}{q_j})^{2}$ for $i,j=1,...,m$.
Hence, from (\ref{eq:Eq201807237841a-SS}), we have
$(\left( q_{i} / \pi_{i} \right) + s - 1 )\frac{q_i}{\pi_i}   =    \pm (\left(  \left( q_{j} / \pi_{j} \right) + s - 1  \right))\frac{q_j}{\pi_j} $,
where it can be concluded that
$q_{i} / \pi_{i}    =    q_{j} / \pi_{j}$
for $i, j=1,...,m$.
Defining $\mu = \frac{ q_{i} }{ \pi_{i} }$ for $i=1,...,m$, from $\sum_{i=1}^{m} q_{i} = \pi_{0}$, we have
$\mu =   \pi_{0} /  \left( \sum_{i=1}^{m} \pi_{i} \right)  $
and thus,
the optimal value of $q_{i}$ is as provided in (\ref{eq:Eq201807278235-SS}).
The acceptable answer
for $s$
is obtained for
$\sum_{i=1}^{m} a_{i}  = 0$, where from (\ref{eq:Eq201807237897a-SS}), we have $ s = 1 - \mu$, and
thus (\ref{eq:Eq201807278283-SS}) is concluded.
For the results provided here for the star topology, the statements in Remark \ref{Remark-StarBranch} are applicable.
Also, if the optimal point is strict, i.e., (\ref{eq:Eq201801262501c}) is satisfied with strict inequality, the convergence rate $(s)$ is equal to
$s = 1 - \mu = 1 - \boldsymbol{P}_{i,0}$.
Note that the results in (\ref{eq:Eq201807237981-SS}), (\ref{eq:Eq201807237991-SS}), (\ref{eq:Eq201807278235-SS}) and (\ref{eq:Eq201807278283-SS}) hold for $m \geq 3$.

\subsubsection*{Case of $m=2$}

In the case of $m=2$, the Star topology reduces to the Path topology and for the equilibrium distributions that satisfy
$\pi_{2}^{2} \geq \pi_{1} \pi_{3}$, the optimal results are same as those provided in Subsection \ref{sec:pathtopology}.
If $\pi_{2}^{2} < \pi_{1} \pi_{3}$, considering
the definition of $\boldsymbol{Z}_{1}$, $\boldsymbol{Z}_{2}$ and $\boldsymbol{Z}_{3}$ and
equations (\ref{eq:Eq20171114938a}) and (\ref{eq:Eq20171114938b})
and defining
$Q_1  =  1 - q_1 \left(  ( 1 / \pi_1 )  +  ( 1 / \pi_2 )  \right)$
and
$Q_2  =  1 - q_2 \left(  ( 1 / \pi_2 )  +  ( 1 / \pi_3 )  \right)$,
we have,
\begin{subequations}
	\label{eq:Eq201801165905-43-SS}
	\begin{gather}
		\left( \left( s - Q_{1} \right) \left( s - Q_{2} \right) - ( q_{1} q_{2} / \pi_{2}^{2} ) \right) a_{1} a_{2}        =  0
		\label{eq:Eq201801165905-43a-SS}
		\\
		\left( \left( s + Q_{1} \right) \left( s + Q_{2} \right) - ( q_{1} q_{2} / \pi_{2}^{2} ) \right) b_{1} b_{2}        =  0
		\label{eq:Eq201801165905-43b-SS}
	\end{gather}
\end{subequations}
From
(\ref{eq:Eq20171114916})
we have
\begin{subequations}
	\label{eq:Eq201801165928-86-SS}
	\begin{align}
		&\left( a_{1} \left( ( 1 / \pi_{1} ) + ( 1 / \pi_{2} ) \right) - ( a_{2} / \pi_{2} ) \right)^{2}
		\label{eq:Eq201801165928-86a-SS}
		\\&
		\quad\quad
		=
		\left( b_{1} \left( ( 1 / \pi_{1} ) + ( 1 / \pi_{2} ) \right) - ( b_{2} / \pi_{2} ) \right)^{2}
		+
		c_{0}^{2}
		\nonumber
		\\&
		\left( a_{2} \left( ( 1 / \pi_{3} ) + ( 1 / \pi_{2} ) \right) - ( a_{1} / \pi_{2} ) \right)^{2}
		\label{eq:Eq201801165928-86b-SS}
		\\&
		\quad\quad
		=
		\left( b_{2} \left( ( 1 / \pi_{3} ) + ( 1 / \pi_{2} ) \right) - ( b_{1} / \pi_{2} ) \right)^{2}
		+
		c_{0}^{2}
		\nonumber
	\end{align}
\end{subequations}
where
the optimal answer is obtained for $b_{1}, b_{2} = 0$, as
reported in (\ref{eq:Eq201808129359-SS}) and (\ref{eq:Eq201808129374-SS}).

The results reported above hold true for clique lifted graphs that their base graph is a path graph with three vertices,
where
the equilibrium distribution $(\pi_i)$ is replaced with sum of the equilibrium distributions $(\pi_{i,j})$ for all vertices of the fiber corresponding to the $i$-th vertex.

\section{Symmetric Tree}
\label{sec:SymmetricTree-Appendix}
Here, the FMRMC problem over symmetric tree topology with symmetric equilibrium distribution is addressed.
This topology is explained in Subsection \ref{sec:SymmetricTree}.

We associate vertex $( 1, \alpha_{1}, ..., \alpha_{k} )$ with vector $\boldsymbol{e}( 1, \alpha_{1}, ..., \alpha_{k} )$ for $k=1, ..., n$, and
root vertex with vector $\boldsymbol{e}(1)$,
where each of these vectors is a column vector of length $N$ with all elements equal to zero, except the element corresponding to their associated vertex.
Based on these vectors, for matrix $\boldsymbol{D}$ we have
$\boldsymbol{D}  =  \pi_{0} \boldsymbol{e}(1) \boldsymbol{e}(1)^{T}  +  \sum_{k=1}^{n} \pi_{k} \sum_{\alpha_{1}, ..., \alpha_{k}} \boldsymbol{e}( 1, \alpha_{1}, ..., \alpha_{k} )  \boldsymbol{e}( 1, \alpha_{1}, ..., \alpha_{k} )^{T}$.
Also, the symmetric Laplacian can be written as
$\boldsymbol{L}\left( q \right)  =  q_{0} \sum_{\alpha_{1}} ( \boldsymbol{e}(1) - \boldsymbol{e}(1, \alpha_{1}) ) ( \boldsymbol{e}(1) - \boldsymbol{e}(1, \alpha_{1}) )^{T}  +  \sum_{k=1}^{n-2}  q_{k} \sum_{\alpha_{1}, ..., \alpha_{k+1}}  ( \boldsymbol{e}( 1, \alpha_{1}, ..., \alpha_{k} ) - \boldsymbol{e}( 1, \alpha_{1}, ..., \alpha_{k+1} ) )   ( \boldsymbol{e}( 1, \alpha_{1}, ..., \alpha_{k} ) - \boldsymbol{e}( 1, \alpha_{1}, ..., \alpha_{k+1} ) )^{T}$.
We define new basis as
$\widetilde{\boldsymbol{e}}( 1, \alpha_{1}, ..., \alpha_{k} )  =  \frac{ 1 }{ \sqrt{m_{0} .... m_{k-1}} }  \sum_{\beta_{1}, ..., \beta_{k-1} = 1}^{m_{0}, ..., m_{k-1}} \omega_{ 1 }^{ \alpha_{1} \beta_{1} }  \omega_{ 2 }^{ \alpha_{2} \beta_{2} } ... \omega_{ k }^{ \alpha_{k} \beta_{k} }  \boldsymbol{e}( 1, \beta_{1}, ..., \beta_{k} )$,
with $\omega_{k} = e^{ \frac{ 2 \pi i }{ m_{k-1} } }$ for $\alpha_{k} = 0, ..., m_{k-1}-1$ and  $k = 1, ..., n-1$.
In the new basis, the transition probability matrix $\boldsymbol{P} = \boldsymbol{I}_{N} - \boldsymbol{D}^{-\frac{1}{2}} \boldsymbol{L}(q) \boldsymbol{D}^{-\frac{1}{2}}$ becomes block diagonal with
blocks $\boldsymbol{P}_{0}$ and $\boldsymbol{P}_{1}$,
where 
$
\boldsymbol{P}_{0}
=
\boldsymbol{I}_{n+1} - \sum_{i=0}^{n-1} q_{i} (  \frac{ \sqrt{ m_{i} } }{ \sqrt{ \pi_{i} } } \boldsymbol{e}^{'}_{i}  -  \frac{1}{ \sqrt{ \pi_{i+1} } } \boldsymbol{e}^{'}_{i+1} )    (  \frac{ \sqrt{ m_{i} } }{ \sqrt{ \pi_{i} } } \boldsymbol{e}^{'}_{i}  -  \frac{1}{ \sqrt{ \pi_{i+1} } } \boldsymbol{e}^{'}_{i+1} )^{T}
$
and
$\boldsymbol{P}_{1}$ is obtained by removing the first row and column of $\boldsymbol{P}_{0}$.
$\boldsymbol{e}^{'}_{i}$ is a $(n+1) \times 1$ vector with $1$ in the $i$-th position and zero elsewhere.
From Cauchy interlacing theorem (Appendix \ref{sec:CauchyInterlacingTheorem}), it can be concluded that the second largest eigenvalue of $\boldsymbol{P}$ is the largest eigenvalue of $\boldsymbol{P}_{1}$, and the smallest eigenvalue of $\boldsymbol{P}$ is
the smallest eigenvalue of $\boldsymbol{P}_{0}$.
Defining vector
$\boldsymbol{\nu} = \gamma^{-1} [ \sqrt{\pi_{0}}, \sqrt{\pi_{1} m_{0}}, ..., \sqrt{\pi_{n} m_{0} ... m_{n-1}} ]^{T}$
with
$\gamma = \pi_{0} + \pi_{1} m_{0} + \pi_{2} m_{0} m_{1} + ... + \pi_{n} m_{0} ... m_{n-1}$,
the FMRMC problem over symmetric tree topology can be written as
\begin{equation}
	\label{eq:Eq201801185995}
	\begin{aligned}
		\min\limits_{q}
		\;\;
		&s,
		\\
		s.t.
		\quad
		&
		\boldsymbol{P}_{1} \preceq s \boldsymbol{I},
		\;\;
		s \boldsymbol{I} \preceq \boldsymbol{P}_{0} -  \boldsymbol{\nu} \boldsymbol{\nu}^{T}.
	\end{aligned}
\end{equation}
Adopting the solution presented in Subsection \ref{sec:pathtopology} and Appendix \ref{sec:pathbranch-Appendix} for Path topology, the optimal results are obtained as in (\ref{eq:Eq201801186054}), (\ref{eq:Eq201801186064}) and (\ref{eq:Eq201801186075}).

\subsection{Symmetric Tree of Depth Two}
\label{sec:SymmetricTreeDepth2-Appendix}
In the following, as an example we provide the optimal $SLEM$ for a symmetric tree topology of depth two with all possible equilibrium distributions.
Defining
$\boldsymbol{Z}_{1}  =  \sum_{k=0}^{1}  a_{k} ( \sqrt{ \frac{ m_{k} }{ \pi_{k} } } \boldsymbol{e}_{k} - \frac{ \boldsymbol{e}_{k+1} }{ \sqrt{k+1} } )$
and
$\boldsymbol{Z}_{2}  =  b_{0} \frac{ \boldsymbol{e}_{1} }{ \sqrt{\pi_{1}} }  +  \sum_{k=1}^{1}  b_{k} ( \sqrt{ \frac{ m_{k} }{ \pi_{k} } } \boldsymbol{e}_{k} - \frac{ \boldsymbol{e}_{k+1} }{ \sqrt{k+1} } )$,
from (\ref{eq:Eq201801185995}), 
we have
\begin{subequations}
	\label{eq:Eq201808146755}
	\begin{gather}
		\left( s - 1 + \frac{ q_{0} }{ \pi_{1} } \right) a_{0} -
		q_{0}
		\psi_{1}
		a_{1}   =   0,
		\label{eq:Eq201808146755a}
		\\
		\left( s - 1 + q_{1}
		\widehat{\pi}_{1}
		\right) a_{1}
		-
		q_{1} \psi_{1}
		a_{0}  =  0,
		\label{eq:Eq201808146755b}
		\\
		\left( s + 1 +
		\psi_{1} q_{0}
		\right) b_{0}  -  q_{0}
		\widehat{\pi}_{0}
		b_1  =  0,
		\label{eq:Eq201808146755c}
		\\
		\left( s + 1 - q_{1}
		\widehat{\pi}_{1}
		\right) b_{1}  +
		q_{1} \psi_{1}
		b_{0}   =  0.
		\label{eq:Eq201808146755d}
	\end{gather}
\end{subequations}
where
$\widehat{\pi}_{i} = ( m_{i} / \pi_{i} ) + ( 1 / \pi_{i+1} )$.
and
$\psi_{1} = \sqrt{m_{1}} / \pi_{1} $.
To obtain non-trivial answers from (\ref{eq:Eq201808146755}), following equations should be satisfied,
\begin{subequations}
	\label{eq:Eq201808146788}
	\begin{gather}
		s^{2}
		- \left(  -2  +
		\frac{ q_{0} }{ \pi_{1} }
		+
		\widehat{\pi}_{1}
		q_{1}  \right) s
		-
		\frac{ m_{1} q_{0} q_{1}  }{ \pi_{1}^{2} }
		=  0,
		\label{eq:Eq201808146788a}
		\\
		s^{2}
		+ \left(  -2 +
		\widehat{\pi}_{0}
		q_{0}
		+
		\widehat{\pi}_{1}
		q_{1} \right) s
		-
		\frac{  m_{1} q_{0} q_{1}  }{  \pi_{1}^{2}  }
		=  0.
		\label{eq:Eq201808146788b}
	\end{gather}
\end{subequations}
In (\ref{eq:Eq201808146755}), if both $a_{i}$ and $b_{i}$ for $i=0,1$ are nonzero, then from (\ref{eq:Eq201808146788a}) from (\ref{eq:Eq201808146788a}), it can be concluded that $q_{0}$ and $q_{1}$ should satisfy
the following,
\begin{equation}
	\label{eq:Eq201808156813}
	\begin{gathered}
		-4 + 2
		\widehat{\pi}_{1}
		q_1 + \left( \frac{m_0}{\pi_0} +  \frac{2}{\pi_1} \right) q_0 = 0.
	\end{gathered}
\end{equation}
\subsubsection*{Case $1$}
For equilibrium distributions that satisfy
$ 2 \pi_{0}\geq m_{0} \pi_{1} $
and
$  m_{0} \pi_{1}^{2} \geq 2 m_{1} \pi_{0} \pi_{2}$,
the optimal value of $q$ is same as those provided in (\ref{eq:Eq201801186054}) and (\ref{eq:Eq201801186064})
and
the optimal value of $SLEM$ is
reported in Subsection \ref{sec:SymmetricTreeDepth2}, for Case $1$.
\subsubsection*{Case $2$}
For equilibrium distributions that satisfy
$2 \pi_{0}< m_{0} \pi_{1}$
and
$m_{0} \pi_{1}^{2} \geq 2 m_{1} \pi_{0} \pi_{2}$,
equation $c_0\neq 0$ leads to $m_0q_0=\pi_0$.
From this equation and (\ref{eq:Eq201808156813}),
the optimal value of $q_{0}$, $q_{1}$ and $SLEM$ are obtained as
reported in Subsection \ref{sec:SymmetricTreeDepth2}, for Case $2$.
\subsubsection*{Case $3$}
For equilibrium distributions that satisfy
$ 2 \pi_{0}\geqslant m_{0} \pi_{1}$
and
$m_{0} \pi_{1}^{2} < 2 m_{1} \pi_{0} \pi_{2}$,
equation $ C_1\neq 0$ leads to $ q_0 +m_1q_1=\pi_1$.
From this equation and (\ref{eq:Eq201808156813}),
the optimal value of $q_{0}$, $q_{1}$ and $SLEM$ are obtained as
reported in Subsection \ref{sec:SymmetricTreeDepth2}, for Case $3$.
\subsubsection*{Case $4$}
For equilibrium distributions that satisfy
$m_{0} \pi_{1} > 2 \pi_{0}$
and
$2 m_{1} \pi_{0} \pi_{2}> m_{0} \pi_{1}^{2}  $,
equations $c_0\neq 0$ and $c_1\neq 0$ lead to $m_0q_0=\pi_0$ and $q_0+m_1q_1=\pi_1$.
Hence optimal value of $q_0$ and $q_1$ are
$q_0=\frac{\pi_0}{m_0}$ and $q_1=\frac{m_0\pi_1-\pi_0}{m_0m_1}$, respectively.
If the resultant optimal $q_0$ and $q_1$ does not satisfy (\ref{eq:Eq201808156813}),
then either $a_{0}, a_{1} = 0$ or $b_{0}, b_{1} = 0$.
From (\ref{eq:Eq20171114916}), it is obvious that $a_{0}, a_{1} = 0$ results in $b_{0}, b_{1} = 0$, which is not acceptable.
Thus, $b_{0}, b_{1} = 0$ and $a_{0}, a_{1} \neq 0$ is the only acceptable case,
and
from (\ref{eq:Eq201808146788a}),
the optimal value of $SLEM$ is obtained as
reported in Subsection \ref{sec:SymmetricTreeDepth2}.
Note that if the resultant optimal $q_0$ and $q_1$ satisfy (\ref{eq:Eq201808156813}), then both equations in (\ref{eq:Eq201808146788}) hold true and thus the above result for optimal $SLEM$ is still true.

\section{Complete Cored Symmetric (CCS) Star}
\label{sec:CCSStar-Appendix}
Here, the FMRMC problem over CCS Star topology with symmetric equilibrium distribution is addressed.
This topology is explained in Subsection \ref{sec:CCSStar}.

We associate vertex $( i, j )$ with vector $\boldsymbol{e}( i, j )$ for $i=1, ..., m$, $j=0, ..., n$,
where each of these vectors is a column vector of length $|\mathcal{V}|$ with all elements equal to zero, except the element corresponding to their associated vertex.
Based on these vectors, for matrix $\boldsymbol{D}$ we have
$\boldsymbol{D}  = \sum_{j=0}^{n} \pi_{j} \sum_{i=1}^{m} \boldsymbol{e}( i, j )  \boldsymbol{e}( i, j )^{T}$.
Also, the symmetric Laplacian can be written as
$\boldsymbol{L}\left( q \right)  =  q_{0} \sum_{ i \neq i^{'}=1 }^{m} ( \boldsymbol{e}(i,0) - \boldsymbol{e}(i^{'},0) ) ( \boldsymbol{e}(i,0) - \boldsymbol{e}(i^{'},0) )^{T}  +  \sum_{j=1}^{n}  q_{j} \sum_{i=1}^{1}  ( \boldsymbol{e}( i, j - 1 ) - \boldsymbol{e}( i, j ) )   ( \boldsymbol{e}( i, j - 1 ) - \boldsymbol{e}( i, j ) )^{T}$.
We define new basis as
$\widetilde{\boldsymbol{e}}( \alpha, j )  =  \frac{ 1 }{ \sqrt{m} }  \sum_{\beta = 1}^{m} \omega^{ i \beta }  \boldsymbol{e}( \beta, j )$,
with $\omega = e^{ i \frac{ 2 \pi }{ m } }$ for $\alpha = 0, ..., m-1$ and $j = 0, ..., n$.
In the new basis,
the operator $\boldsymbol{I} - \boldsymbol{D}^{-\frac{1}{2}} \boldsymbol{L}(q) \boldsymbol{D}^{-\frac{1}{2}}$ is transformed into the block diagonal form,
with $i$-th block equal to 
$\boldsymbol{I}  -  \boldsymbol{D}_{r}^{-\frac{1}{2}} \boldsymbol{L}_{i}(q)  \boldsymbol{D}_{r}^{-\frac{1}{2}}$,
where
$\boldsymbol{D}_{r}  =  diag( \pi_{0}, \pi_{1}, ..., \pi_{n} )$,
and
$\boldsymbol{L}_{0}(q)  =  \sum_{i=1}^{n} q_{i} ( \boldsymbol{e}_{i} - \boldsymbol{e}_{i+1} ) ( \boldsymbol{e}_{i} - \boldsymbol{e}_{i+1} )^{T}$
and $\boldsymbol{L}_{i}(q) = \boldsymbol{L}_{0}(q) + m q_0 \boldsymbol{e}_{1} \boldsymbol{e}_{1}^{T}$,
for $i=1, ..., m-1$.
$\boldsymbol{e}_{i}$ is $(n+1) \times 1$ column vector with one in the $i$-th position and zero elsewhere.
Considering above relation between $\boldsymbol{L}_{0}(q)$ and $\boldsymbol{L}_{i}(q)$ for $i=1, ..., m-1$ and using the Courant-Weyl inequalities theorem \cite{CourantWeylBook1980},
the following can be stated between the eigenvalues of $\boldsymbol{L}_{0}(q)$ and $\boldsymbol{L}_{i}(q)$,
\begin{equation}
	\label{eq:Eq201711182008}
	\begin{gathered}
			\lambda_{n} \left(   \boldsymbol{L}_{i}(q)   \right)
			\leq
			\lambda_{n} \left(   \boldsymbol{L}_{0}(q)   \right)
			\leq
			\cdots
			\leq
			\lambda_{1} \left(   \boldsymbol{L}_{i}(q)   \right)
			\leq
			\lambda_{1} \left(   \boldsymbol{L}_{0}(q)   \right)
			=
			1
	\end{gathered}
\end{equation}
From (\ref{eq:Eq201711182008}), it can be concluded that the second largest and the smallest eigenvalues of
$\boldsymbol{I}  -  \boldsymbol{D}^{-\frac{1}{2}} \boldsymbol{L}(q)  \boldsymbol{D}^{-\frac{1}{2}}$
are the largest and smallest eigenvalues of
$\boldsymbol{I}  -  \boldsymbol{D}_{r}^{-\frac{1}{2}} \boldsymbol{L}_{i}(q)  \boldsymbol{D}_{r}^{-\frac{1}{2}}$.
Hence, the problem reduces to the FMRMC problem over a Path topology with a self loop
on one end
with
$m q_0$.
Thus, the recursive solution presented in Section \ref{sec:pathtopology} for Path topology can be adopted,
which results in the optimal
results
$q_{j} = \frac{ \pi_{j} \pi_{j-1} }  { \pi_{j} + \pi_{j-1} }$ for $j=1, ..., n$
and $q_{0} = \pi_{0} / m$, given that $(m-1) \pi_{1} \leq \pi_{0}$ and the constraint (\ref{eq:Eq201711171284}) are satisfied.

\bibliography{ArXiv_Format}

\end{document}